\documentclass[a4paper, 12pt]{article}

\usepackage[sort&compress]{natbib}
\bibpunct{(}{)}{;}{a}{}{,} 
\RequirePackage[colorlinks,citecolor=blue,urlcolor=blue]{hyperref}

\usepackage{amsthm, amsmath, amssymb, mathrsfs, multirow, url, subfigure}
\usepackage{graphicx} 
\usepackage{ifthen} 
\usepackage{amsfonts}
\usepackage[usenames]{color}
\usepackage{fullpage}
\usepackage{tikz}
\usepackage{float}
\usepackage{bbm}
\usepackage{booktabs}
\usepackage{soul}
\allowdisplaybreaks

\theoremstyle{plain} 
\newtheorem{thm}{Theorem}

\newtheorem{prop}{Proposition}
\newtheorem{lem}{Lemma}

\theoremstyle{definition}

\theoremstyle{remark}

\newcommand{\prob}{\mathsf{P}}

\newcommand{\unif}{{\sf Unif}}

\newcommand{\gam}{{\sf Gamma}}

\newcommand{\chisq}{{\sf ChiSq}}

\newcommand{\betadist}{{\sf Beta}}
\newcommand{\IH}{{\sf IH}}

\newcommand{\ZZ}{\mathbb{Z}}

\newcommand{\TT}{\mathbb{T}}

\newcommand{\iid}{\overset{\text{\tiny iid}}{\,\sim\,}}

\newcommand{\prior}{\mathsf{Q}}

\newcommand{\lPi}{\underline{\Pi}}
\newcommand{\uPi}{\overline{\Pi}}

\begin{document}

\title{\bf 
Valid and efficient possibilistic fusion}
\author{Leonardo Cella \hspace{.2cm}\\
Department of Statistical Sciences, Wake Forest University\\}
\date{}
\maketitle
 
\begin{abstract} 

Besides the classical motivation of fusing evidence from multiple sources, modern inferential procedures based on randomization, resampling, and data splitting often introduce analyst-generated multiplicity, where aggregating outputs across random realizations can improve robustness and stability. This emphasizes the importance of developing principled strategies for fusing measures of evidence across different inferential settings, while preserving the key properties of the adopted inferential framework. The present paper addresses this problem in the context of inferential models (IMs), a possibilistic approach for provably valid statistical inference. Although the fusion of possibility measures has been extensively studied in the possibility-theory literature, existing methods do not, in general, preserve IM validity. We propose a general validity-preserving framework for possibilistic fusion, motivated by the ranking--validification construction underlying IMs. We study the implementation of this framework under independence, arbitrary dependence, and exchangeability of the available IMs, thereby providing a unified approach for IM fusion across a broad range of practically relevant scenarios. The proposed framework also reveals important efficiency considerations, showing that intuitive and commonly used fusion operators may become inefficient in the IM context, so that alternative choices can sometimes be advantageous, including ones that might not appear natural from a purely intuitive standpoint.


\end{abstract}

\noindent%

\section{Introduction}
\label{s:intro}

The problem of combining multiple sources of evidence to better distinguish signal from noise has long been a central contribution of statistics to scientific inquiry. A classical example arises when several independent studies are conducted to address a common question. Each study produces a measure of evidence about the quantity of interest, and the goal is to combine these into a single, more informative summary. In such settings, the multiplicity of evidence arises naturally.

A related—but importantly different—situation arises in modern inference, where randomization, resampling and data splitting have become routine tools for improving robustness and computational tractability. However, these procedures can also introduce inferential variability: different random realizations may yield noticeably different conclusions. This suggests that aggregating multiple such inferential outputs can serve as a natural stabilization principle, which underlies methods such as bootstrap aggregation (bagging) \citep{Breiman1996} and split-and-aggregate procedures \citep{Meinshausen2009}. Note that, in contrast to the classical setting above, the multiplicity of outputs here does not arise naturally from the scientific problem itself, but is instead deliberately introduced by the analyst, while still leading to the same fundamental question of how to ideally combine measures of evidence.

Along these lines, the present paper is primarily focused on exploring new developments in the combination or fusion of {\em inferential models} (IMs) \citep{martinbook}. IMs, ``one of the original statistical innovations of the 2010s'' \citep{cui.hannig.im}, provide a ``best of both worlds'' (frequentist and Bayesian) framework for statistical inference, in that their uncertainty quantification takes the form of both procedures with guaranteed error-rate control and calibrated degrees of belief about assertions concerning the unknown quantity of interest. One factor that makes the reliability of these two forms of uncertainty quantification possible is that IMs are formulated in the language of {\em imprecise probability theory} or, more specifically, {\em possibility theory} \citep[e.g.,][]{dubois.prade.book,dubois2006possibility}. This possibilistic nature of IMs entails the existence of a {\em possibility contour}, which serves as the foundation for all IM-based uncertainty quantification. These possibility contours constitute the inferential objects to be fused throughout the paper.

There exists a vast literature addressing the fusion of possibility contours; see, e.g., \citet{Dubois1999} and the references therein, as well as \citet{DuboisPrade88,dubois:hal-01484952,Dubois2001}. IM contours are, however, a special type of possibility contour, in that they satisfy a key validity property. This is another factor that ensures the reliability of all IM-based uncertainty quantification. Existing fusion strategies do not account for this requirement and therefore do not, in general, preserve validity. This motivates the need for alternative fusion methods. 
The present paper proposes a general validity-preserving approach to IM fusion that accommodates the different multiplicity scenarios described above, enabling IM users to combine IMs in a principled and flexible way across a wide range of settings.

The proposed IM fusion approach is guided by a central feature of the IM framework, namely its two-step construction through {\em ranking} and {\em validification}; see \citet{ryanpp1,imreview}. In the usual IM setting, candidate values of the unknown quantity of interest are first ranked according to their compatibility with the observed data, and this ranking is then validified to produce a valid possibility contour. This same logic is well suited to the fusion problem considered here. Indeed, once multiple IM contours are available, their aggregation through a chosen operator can be viewed as producing a new combined ranking, which must then be appropriately validified in order to preserve the reliability guarantees that distinguish IM-based uncertainty quantification. The main contribution of this paper is to develop validification strategies for such fused rankings under different scenarios, depending on what additional assumptions are made about the available IM contours.

The remainder of the paper is organized as follows. After providing background on IMs in Section~\ref{s:back}, we review in Section~\ref{ss:naive} established fusion strategies from the possibility literature and explain why they are naïve for our purposes, i.e., why they fail to preserve IM validity. We then introduce a general {\em FVN (fuse–validify–normalize)} strategy for validity-preserving IM fusion in Section~\ref{ss:FVN} and study its implementation under three distinct dependence regimes for the available contours: independence (Section~\ref{S:indep}), arbitrary dependence (Section~\ref{s:dependent}), and exchangeability (Section~\ref{s:exch}). The independent setting is the one in which validification most closely resembles that in the general IM construction described in Section~\ref{s:back}, since the sampling distribution of the combined ranking can be derived. 
This approach breaks down in the arbitrarily dependent case, which calls for alternative strategies. We propose one strategy based on suitably upper bounding the aforementioned sampling distribution, and another inspired by recent proposals for combining $p$-values; see, e.g., \cite{vovk2020combining,vovk2022admissible}. Efficiency considerations and recommendations regarding the resulting validifications are also discussed. While these alternative validification strategies also underpin the exchangeable case, the additional structural assumption can be leveraged to enable more efficient IM fusion in some situations. Illustrations are presented in Section~\ref{s:Examples}, and a concise summary is provided along with some future directions in Section~\ref{s:Conclusion}.

Related work in the IM literature includes the recent validity-preserving strategy based on order-statistic fusion operators for dependent IMs proposed by \citet{LinIMAggregation}, as well as the work of \citet{CellaFusion}, which can be viewed as the seed for the developments presented here, although only independent IMs were considered in that work. The proposed FVN strategy substantially generalizes these approaches, both in that it is not tied to specific fusion operators and in the scope of dependence structures it accommodates, encompassing independence, arbitrary dependence, and exchangeability.  \citet{SyringMartin19} investigated the validity and efficiency of different fusion rules in the context of independent IMs. Another contribution of the present work is in this same spirit, namely, to explore—across the aforementioned range of dependence structures—how validity and/or efficiency considerations may rule out commonly used fusion operators in the possibility-fusion literature or, alternatively, motivate moving beyond them, including choices that might not appear natural from a purely intuitive standpoint.



\section{Background on IMs}
\label{s:back}

Frequentist approaches quantify uncertainty through calibrated procedures that control error rates, such as confidence sets and hypothesis tests. Bayesian approaches, by contrast, deliver probabilistic degrees of belief via posterior distributions. Both perspectives offer valuable but incomplete tools for scientific inference, and each has elements the other lacks. This has motivated several unifying frameworks—including fiducial \citep{fisherfiducial}, generalized fiducial \citep{MainHaning}, default-prior Bayes \citep{berger.objective.book}, and confidence distributions \citep{mainconfdist}—which seek to produce prior-free, data-driven belief assignments that are frequentistically calibrated.

More specifically, let $Z = (Z_1, \ldots, Z_n)$, $Z \in \ZZ$, denote the observable data which is modeled by $\{\mathbb{P}_\theta: \theta \in \TT\}$, where $\mathbb{P}_\theta$ is a probability distribution supported on $\ZZ$ that is associated with some unknown quantity $\theta \in \TT$. This setup is intended to be general, encompassing scenarios where $\mathbb{P}_\theta$ conforms to a parametric model, as well as instances where the distribution of $Z$ remains unspecified, with $\theta$ serving to characterize a feature of interest  of this distribution. We assume there's a ``true'' value of the unknown quantity, denoted by $\Theta$, so that $Z \sim \mathbb{P}_\Theta$.  The present goal is to quantify uncertainty about $\Theta$, given $Z=z$.  Those unifying approaches mentioned in the previous paragraph proceed by identifying a probability distribution $\prior_z$ supported on $\TT$ and, for any $H \subseteq \TT$, interpret $\prior_z(H)$ as the $z$-dependent degree of belief in the hypothesis ``$\Theta \in H$.'' Of course, with such an interpretation, it would be undesirable if $\prior_Z(H)$ tended to be large for some $H \not\ni \Theta$ as a function of $Z \sim \mathbb{P}_\Theta$.  Then the following calibration condition makes sense: 
\begin{equation}\label{eq:NeceVal}
\sup_{\theta \not\in H} \mathbb{P}_\theta \{ \prior_{Z}(H) \geq 1-\alpha \} \leq \alpha, \quad  \text{for all $\alpha \in [0,1]$ and all $H \subseteq \TT$.}
\end{equation}
In words, \eqref{eq:NeceVal} states that assigning large belief to false hypotheses about $\Theta$ is a rare event, thereby justifying the inference ``$\Theta \in H$'' when $\prior_z(H)$ is sufficiently large.  Unfortunately for probabilism, the {\em false confidence theorem} \citep{Ryansatellite} implies that there are no data-driven precise probabilities $\prior_z$ that satisfy \eqref{eq:NeceVal}.  Therefore, in order for uncertainty quantification to be reliable in this sense, it's necessary that it take the form of an {\em imprecise probability}. 

Imprecise probabilistic approaches to statistical inference, e.g., Dempster--Shafer theory \citep{dempster.copss, dempster1967, dempster1968a, DEMPSTER2008365, shafer1976mathematical}, belief functions \citep{denoeux.li.2018, denoeux2014} and Inferential Models (IMs) \citep{martinbook}, produce a data-dependent imprecise probability, i.e., a mapping $z \mapsto (\lPi_z, \uPi_z)$, consisting of a lower probability $\lPi_z$ and an upper probability $\uPi_z$, interpreted as degrees of belief and plausibility, respectively.  Both are supported on $\TT$, with the basic property that $\lPi_z(H) \leq \uPi_z(H)$ for all $H$.  The lower and upper probabilities themselves are monotone set functions from $2^\TT$ to $[0,1]$, but they are non-additive and hence can't be called probability measures.  The two are linked via the relation, $\lPi_z(H) = 1 - \uPi_z(H^c)$. But imprecision alone is not sufficient to guarantee that $\lPi_z$ satisfies \eqref{eq:NeceVal}. A special construction is needed and, to the best of our knowledge, IMs are the only framework that achieves this.

IMs have a possibilistic nature. More specifically, they are formulated in the language of possibility theory, the simplest/most-structured form of imprecise probability. This implies the existence of a {\em possibility contour}, namely, a function $\pi_{z}(\theta)$ on $(\mathbb{Z} \times \TT) \rightarrow [0,1]$ satisfying
\begin{equation}\label{eq:consonance}
 \sup_{\theta \in \TT} \pi_{z}(\theta)=1 \quad \text{for all } z.   
\end{equation}
This possibility contour serves as the foundation for all lower/upper probabilities computations through
$\uPi_{z}(H) = \sup_{\theta \in H} \pi_{z}(\theta)$.
Critical to the IM's reliability is the calibration property \eqref{eq:NeceVal} and, of course, certain constraints must be imposed on $\pi_z$ such that this is achieved. It turns out that the relevant constraint is that the random variable $\pi_Z(\Theta)$, as a function of $Z \sim \prob_\Theta$, be like a p-value, i.e., stochastically no smaller than $\unif(0,1)$:
\begin{equation}\label{eq:IMvalidity}
\mathbb{P}_\Theta\{ \pi_{Z}(\Theta) \leq \alpha\} \leq \alpha, \quad  \text{for all $\alpha \in [0,1]$.}
\end{equation}
We refer to this property as {\em validity}. An IM is termed {\em exact} valid when $\pi_{Z}(\Theta) \sim \unif(0,1)$. 

IM validity has several important consequences. First, besides ensuring that $\lPi_z$ satisfies the calibration property in \eqref{eq:NeceVal}, it also guarantees that $\uPi_z$ satisfies 
\[\sup_{\theta \in H} \mathbb{P}_\theta\{ \uPi_Z(H) \leq \alpha\} \leq \alpha \quad  \text{for all $\alpha \in [0,1]$ and all $H \subseteq \TT$}.\] 
In fact, the IMs lower and upper probabilities satisfy even stronger uniform-in-$H$ calibration properties, but that's beyond our present scope; see \citet{CellaMartinSevere}. Second, validity implies that the so-called $\alpha$-cuts of the contour,
\[ C_\alpha(z) = \{ \theta: \pi_z(\theta) > \alpha \}, \quad \alpha \in [0,1], \]
are confidence sets in the sense that $\sup_{\theta \in \TT} \mathbb{P}_\theta\{ C_\alpha(Z) \not\ni \theta\} \leq \alpha$ for each $\alpha$.  Third, in formal decision-making—i.e., finding actions that minimize a suitable upper expected loss—validity implies reliability guarantees for the IM-derived actions \citep{imdec.isipta}.

We point out that, although $\pi_Z(\Theta)$ satisfies the $p$-value-like property in \eqref{eq:IMvalidity}, it is not merely a $p$-value. More precisely, $\pi_z$ is not just a $p$-function, but a special type of $p$-function, since $p$-functions are not required to attain a maximum value of $1$ on $\TT$. This distinction is important because it is precisely this constraint that makes $\uPi_z$ a coherent upper probability in the sense of \citet{walley1991}; see also \citep{miranda.cooman.chapter,lower.previsions.book}.
In particular, $\pi_z$ determines a non-empty collection of data-dependent probability distributions dominated by $\uPi_z$, each of which can be interpreted as a confidence distribution. The possibilistic framework built on the pair $(\lPi_z,\uPi_z)$ is therefore fully compatible with conditional Bayesian reasoning, while at the same time extending it beyond the additive setting. Importantly, these Bayesian-like fixed-data interpretations coexist with the frequentist calibration and error-rate control guarantees that characterize IMs. For strong results on certain elements of this dominated class, including exact probability matching, Bernstein--von Mises behavior, and their implications for efficient computation within the IM framework, see \citet{reimagined,immc}.

How are IMs constructed? More specifically, how can we construct a data-dependent possibility contour that satisfies \eqref{eq:IMvalidity}?  Following \citet{ryanpp1,imreview}, we'll focus on a two-step construction: 
\begin{enumerate}
    \item {\em Ranking-step:} Specify a ranking function $R_{z}(\theta)$ that ranks candidate values $\theta \in \TT$ in terms of how compatible it is with the observed data $z$, with higher values of the ranking function meaning greater compatibility.

    \item {\em Validification-step:} Perform
    \begin{equation}\label{eq:IMvalidification}
\pi_{z}(\theta) = \mathbb{P}_\theta\{R_Z(\theta) \leq R_z(\theta)\}, \quad \theta \in \TT,
\end{equation}
so that the chosen ranking function is transformed into a possibility contour that is valid in the sense of \eqref{eq:IMvalidity}.
\end{enumerate}

As an example, the relative likelihood constitutes a principled choice for the ranking function when $\mathbb{P}_\theta$ is a parametric model. The relative likelihood is itself a data-dependent possibility contour, and inferences based on it have been extensively studied \citep[e.g.,][]{shafer1982, wasserman1990b, denoeux2006, denoeux2014}. However, there is no guarantee that it satisfies the validity property in \eqref{eq:IMvalidity}, so the validification step is performed within the IM framework to ensure reliability of the resulting inferences. For examples in semiparametric, nonparametric, and prediction settings \citep[see][]{imreview,CELLA2024IJAR,cella2025}.

\section{Fusing IMs}

\subsection{Notation and objective}

For ease of notation, throughout the remainder of the paper we suppress the explicit dependence of IM contours on the observed data introduced in Section~\ref{s:back}, and let $\pi_1, \ldots, \pi_K$ represent the data-dependent possibility contours of $K$ IMs, with $K \geq 2$, each formulated to provide reliable inference on the same unknown quantity $\Theta \in \TT$. An illustration for $K=3$ and one-dimensional $\Theta$ appears in Figure~\ref{fig:IM_contours}. The random-variable counterpart of $\pi_k$ will be denoted by $\Pi_k$, and for our purposes it suffices to focus on the case in which the $K$ IMs are exactly valid, i.e., $\Pi_k(\Theta) \sim \unif(0,1)$ for $k = 1, \ldots, K$.


The goal is to fuse these $K$ IM contours so that the result is itself an IM contour. That is, the fused output must remain a possibility contour, attaining a maximum value of $1$ on $\TT$ as in \eqref{eq:consonance}, and must also be valid, satisfying an IM-validity-type property analogous to \eqref{eq:IMvalidity}. The motivation for fusing the $K$ contours depends on the context. For example, in a meta-analysis setting where each IM arises from a different study, fusion is naturally well motivated and aims to produce more informative inference about $\Theta$. A different situation arises when the $K$ IMs are obtained from a construction based on data splitting, in which different splits can yield substantially different contours. In this case, fusion acts as a stabilization device. Illustrations of both settings are provided in Section~\ref{s:Examples}.

\begin{figure}[t]
\begin{center}
\scalebox{0.5}{\includegraphics{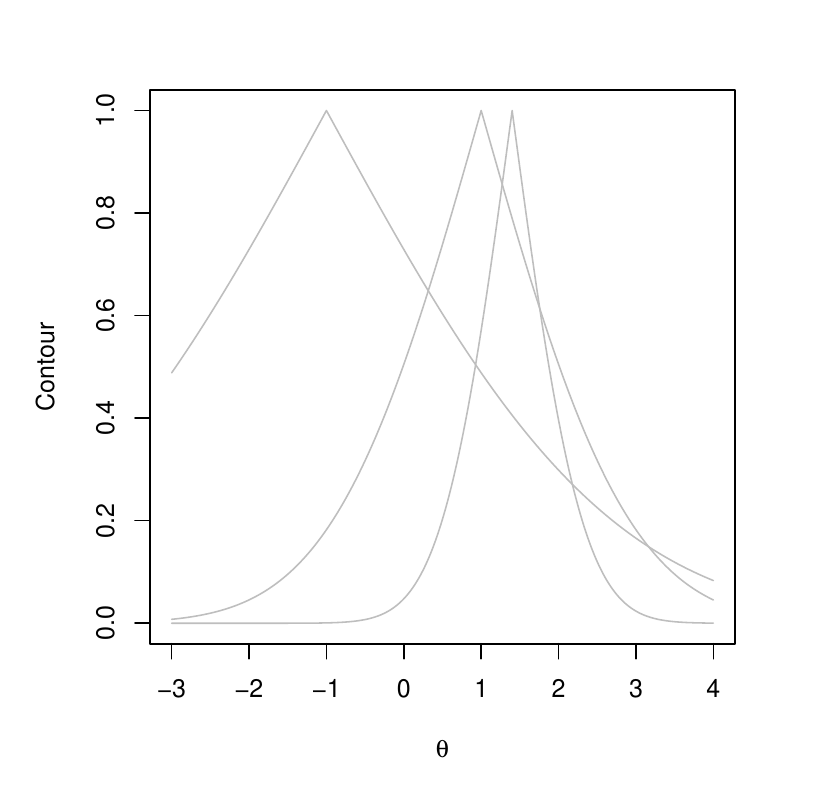}}
\end{center}
\caption{IM contours to be fused.}
\label{fig:IM_contours}
\end{figure}

\subsection{A naïve fusion strategy}
\label{ss:naive}
There is a rich body of literature on the fusion of possibility contours, with a comprehensive overview provided in \citet{Dubois1999}; see also \citet{DuboisPrade88, dubois:hal-01484952, Dubois2001}. In this broader context, the goal is to fuse multiple possibility contours into a single contour in order to identify the most plausible values of an unknown quantity. Fusion methods are typically classified as either \emph{logical} or \emph{statistical}. Logical methods include \emph{conjunctive} and \emph{disjunctive} modes: conjunctive fusion is appropriate when all contours are considered reliable, whereas disjunctive fusion is used when at least one (unknown) contour may be unreliable. Statistical fusion methods, on the other hand, are used when the contours arise from stochastic procedures.

Given that IMs are represented through possibility contours, a natural starting point is to apply these fusion methods developed in the possibility literature. Since each IM contour is assumed to be valid—and therefore reliable—conjunctive fusion is a natural choice. At the same time, because IMs arise from stochastic procedures, statistical approaches are also relevant. We therefore consider both types. 

Let $\gamma_f(\theta)$ denote the result of applying a fusion operator $f$ to the $K$ IM contours at a candidate value $\theta \in \TT$, i.e.,
\[\gamma_f(\theta) = f\big(\pi_1(\theta), \ldots, \pi_K(\theta)\big),\]
and let $\Gamma_f(\theta)$ denote its random-variable counterpart. In the possibilistic framework, interpretability of the fusion operator is crucial for practical adoption. Following \citet{Dubois2001,dubois:hal-01484952}, the primary conjunctive fusion operators are the {\em minimum}, {\em product}, and {\em linear-product}, while the main statistical fusion operators are the {\em average} and {\em geometric average}.

A purely logical view of the fusion process arises when the minimum operator is used, leading to 
$\gamma_{\mathrm{min}}(\theta) = \min\{\pi_1(\theta), \ldots, \pi_K(\theta)\}$,
where the IM assigning the smallest possibility degree to $\theta$ is regarded as the ``best informed'' with respect to this value. The product-based fusion $\gamma_{\mathrm{prod}}(\theta) = \prod_{k=1}^K\pi_k(\theta)$
yields possibility degrees strictly smaller than those produced by the minimum rule, thereby reinforcing the lack of full plausibility. The fusion based on the linear-product operation is given by 
\[\gamma_{\mathrm{l.prod}}(\theta) = \max\left\{0,\sum_{k=1}^K\pi_k(\theta) - K + 1\right\}.\] 
This is considered a drastic fusion rule that discards values of $\theta$ deemed insufficiently plausible by all IMs. The average $\gamma_{\mathrm{avg}}(\theta) = \frac{1}{K} \sum_{k=1}^K \pi_k(\theta)$
is appropriate when a central summary of the $K$ contours is desired, while the geometric average 
\[\gamma_{\mathrm{g.avg}}(\theta) = \left(\prod_{k=1}^K \pi_k(\theta)\right)^{1/K}\] introduces a conjunctive flavor, since any value assigned zero possibility by one IM is also ruled out by the fused contour.

An application of these five fusion operators is illustrated in Figure~\ref{fig:Naive_fusion}(a) for the $K=3$ contours shown in Figure~\ref{fig:IM_contours}. Notably, the resulting fused contours do not qualify as possibility contours, since they fail to attain a maximum value of one. This issue has been recognized in the possibility fusion literature, where a normalization step is typically recommended: $\gamma_f(\theta)$ is divided by $\sup_{\vartheta \in \TT} \gamma_f(\vartheta)$.
Figure~\ref{fig:Naive_fusion}(b) displays the five normalized fused contours. Note how the characteristic features of each fusion operator appear here, e.g., the drastic nature of the linear-product and the conjunctive flavor of the geometric average.

\begin{figure}[t]
\begin{center}
\subfigure[]{\scalebox{0.5}{\includegraphics{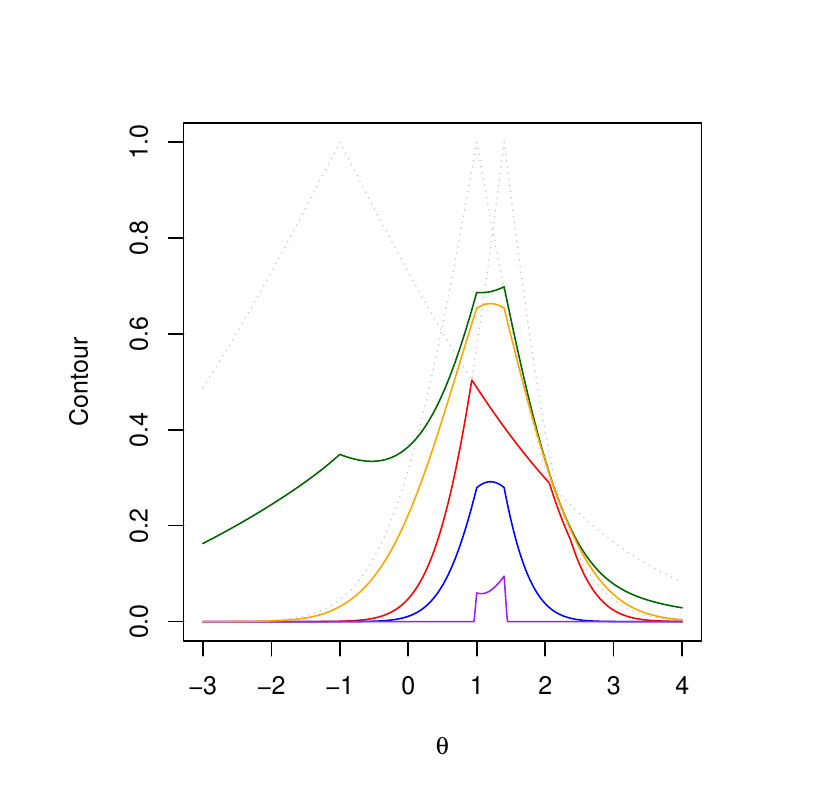}}}
\subfigure[]{\scalebox{0.5}{\includegraphics{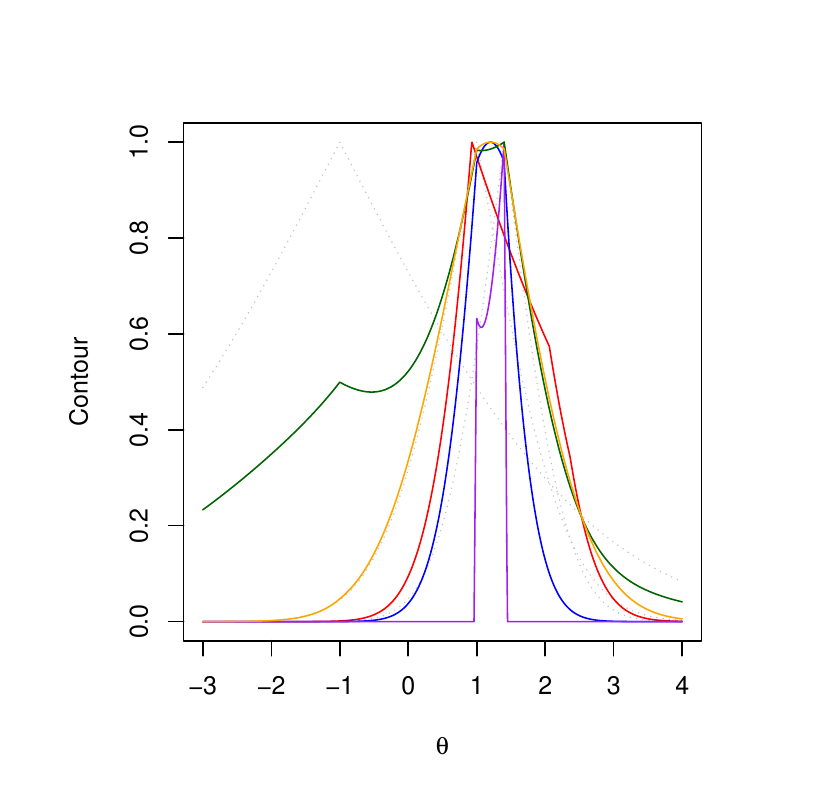}}}
\end{center}
\caption{Minimum (red), product (blue), linear-product (purple), average (green), and geometric average (orange) fusion operators applied to the three IM contours (grey) are shown in Panel~(a); their normalized versions appear in Panel~(b).}
\label{fig:Naive_fusion}
\end{figure}

At first glance, it may seem that the fusion process is complete and the normalized $\gamma_f$ can be used for inference on $\Theta$. 
However, although the $K$ IM contours we started with are valid—each satisfying \eqref{eq:IMvalidity}—this fusion approach does not guarantee that the resulting contour retains this property. In other words, the fusion methods commonly adopted in the possibility literature do not account for calibration guarantees. This is not surprising, since validity is not a general feature of possibility contours, but rather a defining requirement of IM-based contours. A validity-preserving fusion strategy is therefore required in the IMs setting.

\subsection{A general validity-preserving fusion strategy}
\label{ss:FVN}

Recall from Section~\ref{s:back} that the ranking function used to measure the compatibility between candidate values $\theta \in \TT$ and observed data $z$ in the ranking step of the IM construction may lack proper calibration, despite its intuitive appeal. This is precisely why the validification in \eqref{eq:IMvalidification} is necessary. The same reasoning applies in our fusion context: we begin with valid possibilistic rankings for $\Theta$ from $K$ different sources and seek to unify them into a single ranking $\gamma_f$ via a chosen fusion operator $f$. Fusion can therefore be viewed as a re-ranking process. However, no matter how intuitive this re-ranking is, there is no guarantee that $\gamma_f$ remains valid—hence the need for a validification step. 

This suggests that fusing IMs amounts to applying the ranking–validification two-step IM construction to a setting where the ``data'' are themselves IM contours, with the ranking step corresponding to the fusion of these contours via an operator chosen by the analyst. However, a third step must be added, namely a normalization step, since in practice it is highly unlikely that $\sup_{\theta \in \TT} \gamma_f(\theta)$ will equal one; see Figure~\ref{fig:Naive_fusion}(a). 
Normalization was not mentioned in the IM construction of Section~\ref{s:back} because it is usually the case that the chosen ranking function 
$R_Z(\theta)$ has a common maximum value that is attained for some $\theta \in \TT$ for every $z \in \mathbb{Z}$. For example, when $R_z(\theta)$ is the relative likelihood, it attains its maximum value of one at the maximum likelihood estimator for each $z \in \mathbb{Z}$. Normalization is nonetheless not unfamiliar in the IM framework; see, for example, Section~7 of \citet{CELLA20221} and \citet{LFIM}.

The proposed three-step strategy for validity-preserving IM fusion, which we will refer to as FVN—fuse–validify–normalize—is as follows:
\begin{enumerate}
    \item {\em Fusion-step:} Fuse the $K$ IM possibility contours through an appropriate fusion operator $f: [0,1]^K \to [0,1]$, typically chosen to be monotone in each coordinate, in order to produce a combined ranking 

    \[\gamma_f(\theta) = f(\pi_1(\theta), \ldots, \pi_K(\theta)), \quad \theta \in \TT.\]

    \item {\em Validification-step:} Validify $\gamma_f$ through
\[\dot\gamma_f(\theta) = \mathbb{P}_\theta\left\{\Gamma_f(\theta) \leq \gamma_f(\theta)\right\}, \quad \theta \in \TT,\]
so that the combined ranking obtained in the fusion-step is valid in the sense that 
\begin{equation}\label{eq:ValDotGamma}
\mathbb{P}_\Theta\{\dot\Gamma_f(\Theta) \leq \alpha\}\leq \alpha, \quad  \text{for all $\alpha \in [0,1]$}.     
\end{equation}

\item {\em Normalization-step:} Normalize $\dot\gamma_f$ through
    \begin{equation}\label{eq:normGeneral}
     \ddot\gamma_f(\theta) = \frac{\dot\gamma_f(\theta)}{\sup_{\vartheta \in \TT} \dot\gamma_f(\vartheta)}, \quad \theta \in \TT,   
    \end{equation}
    so that the valid combined ranking obtained in the validification-step becomes a valid possibility contour.
\end{enumerate}

\begin{thm}
The FVN strategy yields a valid fused possibility contour in the sense that
\[\mathbb{P}_\Theta\{\ddot\Gamma_f(\Theta) \leq \alpha\}\leq \alpha, \quad  \text{for all $\alpha \in [0,1]$}.\]    
\end{thm}
\begin{proof}
        Let $G$ be the cumulative distribution function of $\Gamma_f(\Theta)$. The probability integral transform ensures that $G(\Gamma_f(\Theta)) = \dot \Gamma_f(\Theta)$ is stochastically no smaller than $\unif(0,1)$, satisfying therefore \eqref{eq:ValDotGamma}. The normalization in \eqref{eq:normGeneral} guarantees that $\ddot\gamma_f$ is a possibility contour as $\sup_{\vartheta \in \TT}\ddot\gamma_f(\vartheta)=1$. Since $\ddot\gamma_f(\theta) \geq \dot\gamma_f(\theta)$ for all $\theta \in \TT$, 
        \[\mathbb{P}_\Theta\{\ddot\Gamma_f(\Theta)\le\alpha\}\;\le\;\mathbb{P}_\Theta\{\dot\Gamma_f(\Theta)\le\alpha\}\;\le\;\alpha,\]
 establishing the result.
\end{proof}

We emphasize once again that the proposed FVN strategy follows the same spirit as the ranking + validification IM construction reviewed in Section~\ref{s:back}. We refer to the ranking step as the fusion step to make clear that we are combining multiple IMs. 
In principle, the fusion and normalization steps are relatively straightforward; indeed, this is essentially what was done in the naïve strategy in Section~\ref{ss:naive}. The primary challenge of the FVN strategy lies in the validification step.

In the remainder of the paper, we explore strategies to address this challenge that depend on what additional assumptions—if any—are made about the $K$ IM contours, beyond the baseline assumption that each is exact valid. Independent IMs are considered in Section~\ref{S:indep}, while arbitrarily dependent IMs are treated in Section~\ref{s:dependent}. Despite the considerable flexibility in choosing the fusion operator $f$, we focus on the five operators introduced in Section~\ref{ss:naive}, namely, the minimum, product, linear-product, average, and geometric average, as these are the most commonly used when fusing possibility contours. We will see, however, that there are compelling reasons to look beyond these in some situations in order to improve efficiency; see the examples in Section~\ref{s:Examples}.

The above approach of obtaining a combined ranking in the fusion step by directly applying a fusion operator $f$ to the $K$ contours is natural and intuitive. However, somewhat surprisingly, we will see in Section~\ref{s:exch} that efficiency considerations call for a modification of this approach when dealing with exchangeable IMs. The fusion operator still plays an important role, but the final fused object does not necessarily correspond to its direct application to the $K$ IMs.


\section{Fusing independent IMs}
\label{S:indep}
Consider the case where the $K$ IMs are independent—for example, because they are built from independent data sources—so that $\Pi_1(\Theta), \ldots$, $\Pi_K(\Theta) \iid \unif(0,1)$. This key property allows us to perform the validification step in the exact same spirit as in Section~\ref{s:back}, as it enables the derivation of the distribution of $\Gamma_f(\Theta)$. 
To this end, we begin by recalling three classical results concerning iid \unif(0,1) random variables.

\begin{lem}\label{lemma:indep}
Let $U_1, \ldots$, $U_K \iid \unif(0,1)$. Then
\begin{itemize}
\item $U_{(k)}\sim\betadist(k,K-k+1)$ for $k=1,\ldots,K$;
\item $-\log \prod_{k=1}^K U_k\sim \gam(K,1)$;
\item $\sum_{k=1}^K U_k\sim \IH(K)$, the Irwin--Hall distribution with parameter $K$ \citep{irwin1927,hall1927}, whose pdf is given by
\begin{equation}\label{eq:IH}
h(x) = \frac{1}{(K - 1)!} \sum_{j=0}^{\lfloor x \rfloor} (-1)^j \binom{K}{j} (x - j)^{K - 1}, \quad \text{for } x \in [0, K].
\end{equation}
\end{itemize}
\end{lem}

\begin{thm}\label{thm:indep}
If the $K$ IMs are independent, so that $\Pi_1(\Theta), \ldots$, $\Pi_K(\Theta) \iid \unif(0,1)$, then the distributions of the minimum, product, linear-product, average, and geometric average fusion operators take, respectively, the following forms:
 \begin{itemize}
     \item $\Gamma_{\mathrm{min}}(\Theta)\sim \betadist(1,K)$;
     \item $\Gamma_{\mathrm{prod}}(\Theta) \stackrel{d}{=} \exp(-V), \quad V \sim \gam(K,1)$;
     \item $Y=\Gamma_{\mathrm{l.prod}}(\Theta)$ has a point mass $H(K - 1)$ at 0, and density
\[g(y) = h(y + K - 1), \quad  y \in (0, 1],\]
where $h$ is the pdf of the $\IH(K)$ distribution in \eqref{eq:IH} and $H$ is the corresponding cdf;
     \item $\Gamma_{\mathrm{avg}}(\Theta)\sim \frac{1}{K}\IH(K)$;
     \item $\Gamma_{\mathrm{g.avg}}(\Theta) \stackrel{d}{=} \exp(-V/K)$, with $V$ as above.

 \end{itemize}
\end{thm}

The proof of the claims in Theorem~\ref{thm:indep} is omitted, since they follow directly from the three results in Lemma~\ref{lemma:indep}. Three remarks are in order. First, although the linear-product operator has an intuitive conjunctive interpretation, it suffers from a serious efficiency problem under independence. Indeed, the point mass at zero, of size $H(K-1)$, increases rapidly with $K$—for example, it is equal to $0.5$, $0.83$, and $0.96$ for $K=2,3,4$, respectively—and converges to one as $K \to \infty$. Consequently, the corresponding validified contour becomes nearly flat even for moderate $K$, making it poorly discriminative and therefore highly inefficient. We therefore do not recommend independent IMs fusion based on the linear-product operator. Second, it is easy to see that the product and geometric-average fusion operators yield identical validified contours, i.e., $\dot\gamma_{\mathrm{g.avg}}(\theta)=\dot\gamma_{\mathrm{prod}}(\theta)$ for all $\theta \in \TT$. The choice between the product and geometric-average operators when the IMs are independent has therefore no impact on the resulting fused IM, despite their seemingly different interpretations. Third, the reader may have noticed similarities to methods proposed for merging independent $p$-values in the statistics literature, e.g., \citet{Fisher32,Pearson1934,cousins2008annotated,OwenPvalue}. This is not surprising, given the connections between $p$-values and IM contours; see Section~\ref{s:back}. These connections will also play an important role in the validification of dependent and exchangeable IMs in the subsequent sections.

As an illustration, consider the case where the $K=3$ contours shown in Figure~\ref{fig:IM_contours} are independent. Figure~\ref{fig:indep_fusion} shows the fused IMs obtained using the FVN strategy, with validification carried out using the distributions in Theorem~\ref{thm:indep}. The extreme inefficiency of the linear-product-based IM is confirmed. 


\begin{figure}[t]
\begin{center}
\scalebox{0.5}{\includegraphics{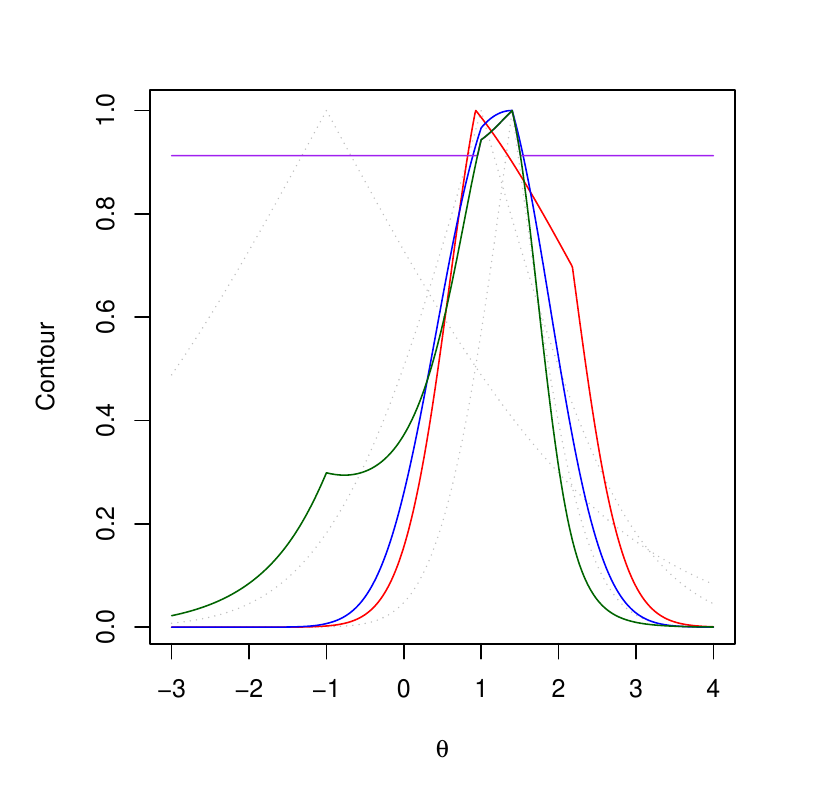}}
\end{center}
\caption{Fused IM contours obtained via the FVN strategy using the minimum (red), product/geometric-average (blue), linear-product (purple), and average (green) fusion operators, applied to the three independent IM contours (grey).}
\label{fig:indep_fusion}
\end{figure}

\section{Fusing dependent IMs}
\label{s:dependent}
We now consider the case where the $K$ IMs may be arbitrarily dependent—for example, when they are constructed using partially shared data, or when no assumptions are made about the nature of their dependence. If all we know about $\Pi_1(\Theta), \ldots, \Pi_K(\Theta)$ is that each is marginally $\unif(0,1)$, then validifying $\gamma_f$ as was done in Section~\ref{S:indep} is infeasible, as this limited information is insufficient to determine the distribution of  $\Gamma_f(\Theta)$. New strategies are required in this setting.

The first strategy we explore consists of deriving an upper bound for the distribution of $\Gamma_f(\Theta)$. Consider $\Gamma_{min}$. An upper bound for its cdf can be easily obtained by the well-known Bonferroni correction:
\begin{align}\label{eq:Bonf}
G_{min}(t) = \mathbb{P}_\Theta\{\Gamma_{min}(\Theta) \leq t\} 
&=\mathbb{P}_\Theta\left\{\bigcup_{k=1}^K \{\Pi_k(\Theta) \le t\}\right\}, \quad t \in [0,1], \nonumber \\
&\leq \sum_{k=1}^K \mathbb{P}_\Theta\{\Pi_k(\Theta) \leq t\} = Kt.
\end{align}

This upper bound suggests that $\dot\gamma_{min}(\theta) = \min\{1,K\gamma_{min}(\theta)\}$ serves as an appropriate validified minimum fusion function in the arbitrarily dependent setting. Validity follows because this construction inflates what the validification would have produced if $G_{min}$ were available. Although this strategy may appear specific to the minimum fusion rule, the underlying idea—upper bounding the cdf of $\Gamma_f(\Theta)$ and basing the validification on that bound—is fully general. 

\begin{prop}
Let $G_f(t) = \mathbb{P}_\Theta\{\Gamma_f(\Theta) \le t\}$ denote the distribution function of $\Gamma_f(\Theta)$. If $G_f(t) \le h(t)$ for every $t \in [0,1]$, then
\[\dot{\gamma}_f(\theta) = \min\{1,\, h(\gamma_f(\theta))\}\]
defines a validified fusion rule in the sense of \eqref{eq:ValDotGamma}.
\end{prop}

The challenge, of course, lies in identifying useful upper bounds. As a first strategy, we note that the minimum bound above can be leveraged to obtain corresponding bounds for the remaining four fusion operators considered here. They are reported in Table~\ref{tab:ub}. The bounds for the average and geometric average follow from the fact that each of these is always greater than or equal to the minimum. The bounds for the product and linear-product come from the observation that $\Gamma_{prod}(\Theta) \leq t$ implies $\Gamma_{min}(\Theta) \leq t^{1/K}$ and $\Gamma_{l,prod}(\Theta) \leq t$ implies $\Gamma_{min}(\Theta) \leq 1 - \frac{1-t}{K}$. Note that the resulting upper bound for the linear-product is trivial, as it is greater than or equal to one for all $t\in(0,1)$.

The minimum-based upper bounds for the product, average, and geometric average fusion operators are nontrivial. However, from an efficiency standpoint, it is natural to ask whether operator-specific bounds can be derived instead of relying solely on the minimum-based ones. 

Towards this, start by noticing the equivalence
\[\Gamma_{min}(\Theta)\leq t \quad \quad \iff\quad  \quad \frac1K\sum_{k=1}^K \mathbbm{1}\{\Pi_k(\Theta) \leq t\} \geq \frac1K. \]
In words, $\Gamma_{min}(\Theta) \leq t$ means that at least one $\Pi_k(\Theta), k=1, \ldots, K$, is at most $t$. This equivalence is particularly useful because the right-hand side formulation is directly suited for applying Markov's inequality:
\[\mathbb{P}_\Theta\left\{\frac1K\sum_{k=1}^K \mathbbm{1}\{\Pi_k(\Theta) \leq t\} \geq \frac1K\right\} \leq K\mathbb{E}\left[\mathbbm{1}\{\Pi_1(\Theta) \leq t\} \right] = Kt,\]
leading to the same upper bound as in \eqref{eq:Bonf}. This same strategy---expressing the event $\Gamma_f(\Theta) \leq t$ in a form that naturally lends itself to an application of Markov’s inequality---is what we use to derive alternative upper bounds for the remaining fusion operators. Specific equivalent events chosen for each fusion operator, along with the resulting bounds, are displayed in Table~\ref{tab:ub}.

\begin{table}[tb]
\caption{Upper bounds for $\mathbb{P}_\Theta\{\Gamma_f(\Theta)\le t\}$, $t \in (0,1)$, based on the minimum fusion rule and on Markov's inequality.}
\renewcommand{\arraystretch}{1.3}
\setlength{\tabcolsep}{6pt}
\centering
\begin{tabular}{c c c c}
\hline
$\gamma_f$ & Min-based bound & $\Gamma_f(\Theta)\le t \iff$ & Markov bound \\
\hline
$\gamma_{\text{prod}}$
& $K t^{1/K}$
& $\displaystyle \frac{1}{K}\sum_{k=1}^K -\ln \Pi_k(\Theta)
   \ge \frac{-\ln t}{K}$
& $\displaystyle \frac{-K}{\ln t}$ \\[6pt]

$\gamma_{\text{l.prod}}$
& $K-1+t$
& $\displaystyle \frac{1}{K}\sum_{k=1}^K 1-\Pi_k(\Theta)
   \ge \frac{1-t}{K}$
& $\displaystyle \frac{K}{2(1-t)}$ \\[6pt]

$\gamma_{\text{avg}}$
& $K t$
& $\displaystyle \frac{1}{K}\sum_{k=1}^K 1-\Pi_k(\Theta)
   \ge 1-t$
& $\displaystyle \frac{1}{2(1-t)}$ \\[6pt]

$\gamma_{\text{g.avg}}$
& $K t$
& $\displaystyle \frac{1}{K}\sum_{k=1}^K -\ln \Pi_k(\Theta)
   \ge -\ln t$
& $\displaystyle \frac{-1}{\ln t}$ \\
\hline
\end{tabular}
\label{tab:ub}
\end{table}

Once again no nontrivial bound is obtained for the linear-product operator since $\frac{K}{2(1 - t)} \ge 1$ for all $t \in (0,1)$. For the remaining fusion operators, note that, when $K=2$, the minimum-based bounds are uniformly tighter than the Markov bounds. When $K>2$, neither bound uniformly dominates the other over $t \in (0,1)$.

The second strategy we explore for the validification of $K$ arbitrarily dependent IMs differs from the first in that, rather than manipulating $\Gamma_f(\Theta) \le t$, we work directly with $\dot{\Gamma}_f(\Theta) \le t$. This approach is closely related to recent developments in the p-value merging literature \citep{vovk2020combining,vovk2022admissible}, so we begin with a brief review.

Let $(\Omega,\mathcal{F},\prob)$ be a probability space. A p-value for testing $\prob$ is a measurable function $p: \Omega \to [0,1]$ that quantifies the compatibility between the observed data in $\Omega$ and $\prob$, with larger (smaller) values indicating greater (less) compatibility. As a random variable under $\prob$ it is stochastically no smaller than \unif(0,1), therefore satisfying
\begin{equation}\label{eq:PVvalidity}
 \prob\{P \leq \alpha\} \leq \alpha, \quad \text{for all } \alpha \in (0,1).
\end{equation}

Suppose $\boldsymbol{p}=p_1, \ldots, p_K$ are $K$ p-values, $K \geq 2$. The goal is to merge these into one p-value. Therefore, a \emph{p-merging} function $W : [0,1]^K \to [0,1]$ satisfies
\[\prob\{W(\boldsymbol{P}) \leq \alpha\} \leq \alpha, \quad \text{for all } \alpha \in (0,1). \]
It is {\em symmetric} if $W(\boldsymbol{p})$ is invariant under any permutation of $\boldsymbol{p}$, and it is {\em homogeneous} if $W(a\boldsymbol{p}) = aW(\boldsymbol{p})$
for all $\boldsymbol{p}$ with $W(\boldsymbol{p})\leq 1$ and $a \in (0,1)$. 
The rejection region at level $\alpha \in (0,1)$ of $W$ is defined as
\[R_\alpha(W) = \{\boldsymbol{p} \in [0,1]^K : W(\boldsymbol{p}) \leq \alpha\}.\] 
Of course, $\prob\{\boldsymbol{P} \in R_\alpha \} \leq \alpha$. As we will see below, rejection regions are a key ingredient in p-value merging techniques.

An e-variable is a non-negative extended random variable $E: \Omega \to [0,\infty]$ with $\mathbb{E}[E]\leq 1$. A calibrator is a decreasing function $g:[0,\infty) \rightarrow [0,\infty]$ satisfying $g=0$ on $(1,\infty)$ and $\int_0^1g(p)dp \leq1$. The purpose of a calibrator is to transform any p-value into an e-value. The motivation is that merging p-values into a p-value is a difficult and context dependent task. On the other hand, merging e-values into an e-value is done by simply averaging them. Moreover, an e-value can be transformed into a p-value by a reciprocal operation. Thus, although e-values have broader roles in modern statistics; see, e.g., \citet{grunwald2024safe,ramdas2023game,RamdasEVbook}; in the p-value merging literature they primarily serve as an intermediate tool.

A key result for our purposes is that valid p-merging rules can be constructed
in terms of calibrators. In particular, \citet{vovk2022admissible} show that,
for any calibrator $g$, the rejection region
\begin{equation}\label{eq:pvRR}
R_\alpha
=
\left\{
\boldsymbol{p} \in [0,1]^K :
\frac{1}{K} \sum_{k=1}^K g\!\left(\frac{p_k}{\alpha}\right) \geq 1
\right\}, \quad \text{for each }\alpha \in (0,1),
\end{equation}
defines a symmetric and homogeneous p-merging function.

But how can the above results help in the fusion of arbitrarily dependent IMs? The key is to note that $\pi_1(\theta), \ldots, \pi_K(\theta)$ can be viewed as $K$ p-value analogues for testing $\Theta = \theta$. This suggests seeking
validified fusion rules whose rejection event $\dot\gamma_f(\theta)\le \alpha$, $\alpha \in (0,1)$,
can be rewritten in the calibrator form
\begin{equation}\label{eq:calibrator}
\frac{1}{K}\sum_{k=1}^K g\!\left(\frac{\pi_k(\theta)}{\alpha}\right)\ge 1
\end{equation}
for some calibrator $g$, since such representations yield valid p-merging functions.

This formulation is, however, too general to be directly useful, unless some
structure is imposed on $\dot\gamma_f$. For example, we can restrict
$\dot\gamma_f$ to a linear validification of the form $c\gamma_f$, with $c>0$.
The strategy is therefore to identify values of $c$ for which the rejection event
$c\gamma_f(\theta) \le \alpha$ admits a representation of the form \eqref{eq:calibrator}
for some calibrator $g$,
leading to the validified fusion operator
\begin{equation}\label{eq:LinearVal}
\dot\gamma_f(\theta) = \min\{1,c\gamma_f(\theta)\}.
\end{equation}
The next theorem establishes these values of $c$ for the minimum, average,
and geometric average fusion operators. The product and linear-product fusion operators cannot be validified through this strategy because, even though they are symmetric, they are not homogeneous. The subsequent theorem shows that a linear validification of these operators is actually impossible.

\begin{thm}
\label{thm:CvalsDep}
For the minimum, average, and geometric-average fusion operators, the choices
$c=K$, $c=2$, and $c=e$, respectively, allow the rejection event
$c\gamma_f(\theta)\le \alpha$, $\alpha \in (0,1)$, to be rewritten in the calibrator form
\eqref{eq:calibrator}. 
\end{thm}

\begin{proof}
The following equivalences follow by straightforward algebra:
\begin{align*}
c \gamma_{\mathrm{min}}(\theta) \leq \alpha
&\iff
\frac1K \sum_{k=1}^K K\mathbbm{1}\left\{\frac{\pi_k(\theta)}{\alpha} \leq \frac1c\right\} \geq 1, \\
c \gamma_{\mathrm{avg}}(\theta) \leq \alpha
&\iff
\frac{1}{K} \sum_{k=1}^K\left(c - c\frac{\pi_k(\theta)}{\alpha}\right) \geq c-1, \\
c \gamma_{\mathrm{g.avg}}(\theta) \leq \alpha
&\iff
\frac{1}{K} \sum_{k=1}^K - \ln\left(\frac{\pi_k(\theta)}{\alpha}\right) \geq \ln c.
\end{align*}
Setting $c=K$, $c=2$, and $c=e$, respectively, yields the calibrator form in
\eqref{eq:calibrator} with
$g(x)=K\mathbbm{1}\{x\le 1/K\}$,
$g(x)=(2-2x)_+$, and 
$g(x)=(-\ln x)_+$.
\end{proof}

\begin{thm}\label{thm:impossible-linear-validification}
For the product and linear-product fusion operators, there is no finite constant $c>0$ such that, for \emph{all} dependence structures among
$\Pi_1(\Theta),\ldots,\Pi_K(\Theta)$, the linear validification in \eqref{eq:LinearVal}
satisfies the calibration property in \eqref{eq:ValDotGamma}.
\end{thm}
\begin{proof}
It is sufficient to find a joint law for
$\bigl(\Pi_1(\Theta),\ldots,\Pi_K(\Theta)\bigr)$ and some $\alpha\in(0,1)$ such that, for any finite $c>0$, \eqref{eq:ValDotGamma} is violated.
Towards this, consider the case where
\[\Pi_1(\Theta)=\cdots=\Pi_K(\Theta)=U,\qquad U\sim \unif(0,1).\]
For the product fusion operator, $\Gamma_{\mathrm{prod}}(\Theta)=U^K$. Let $\alpha < c$, so
\[\mathbb{P}_\Theta\{\dot\Gamma_{\mathrm{prod}}(\Theta)\le \alpha\}
=\mathbb{P}_\Theta\{cU^K\le \alpha\}
=\bigl(\alpha/c\bigr)^{1/K}.\]
For any $\alpha\in\bigl(0,\min\{c,\,c^{-1/(K-1)}\}\bigr)$, $\bigl(\alpha/c\bigr)^{1/K}>\alpha$, therefore violating \eqref{eq:ValDotGamma}. Now, for the linear-product fusion operator, 
\[\Gamma_{\mathrm{l.prod}}(\Theta)=\max\{0, KU-(K-1)\},\]
so $\Gamma_{\mathrm{l.prod}}(\Theta)=0$ whenever $U\le (K-1)/K$. Consequently,
\[
\mathbb{P}_\Theta\{\dot\Gamma_{\mathrm{l.prod}}(\Theta)\le \alpha\}
\ge \mathbb{P}_\Theta\{c\Gamma_{\mathrm{l.prod}}(\Theta)=0\}
=\frac{K-1}{K}.
\]
Therefore, for any $\alpha<(K-1)/K$, \eqref{eq:ValDotGamma} is violated.
\end{proof}

Unsurprisingly, the p-value–based validification of the minimum fusion operator coincides with its Markov-based counterpart. For the average and geometric-average operators, however, the Markov-based validifications are nonlinear, so the p-value–based strategy naturally produces different results. These also differ from the minimum-based validifications, except for the average operator when $K=2$, where the two approaches coincide. From an efficiency standpoint, the p-value–based validification is uniformly the most efficient for the average operator. The same holds for the geometric-average operator when $K\geq 3$; when $K=2$, however, the minimum-based strategy yields the most efficient validification.

None of the techniques considered in this section provide a meaningful validification for the linear-product fusion operator under arbitrary dependence, and we therefore do not recommend its use in this setting. Although a linear validification is impossible for the product operator, nonlinear validifications based on the minimum bound or Markov’s inequality remain applicable. Table~\ref{tab:RecDep} summarizes the recommended validifications for each fusion operator under arbitrary dependence, based on the preceding efficiency considerations. Figure~\ref{fig:dep_fusion} displays the fused IM contours obtained using these recommendations for the $K=3$ contours shown in Figure~\ref{fig:IM_contours}, now assumed dependent (with $-K/\ln \gamma_{\mathrm{prod}}$ used for the product).

\begin{table}[tb]
\caption{Recommended $h(\gamma_f)$ in the validification 
$\dot{\gamma}_f(\theta)=\min\{1,h(\gamma_f(\theta))\}$ 
for $K$ arbitrarily dependent IM contours fused under the fusion operator $f$.}
\centering
\renewcommand{\arraystretch}{1.2}
\setlength{\tabcolsep}{12pt}
\begin{tabular}{c c c}
\hline
$\gamma_f$ & $K=2$ & $K>2$ \\
\hline

$\gamma_{\mathrm{min}}$
& $2\gamma_{\mathrm{min}}$
& $K\gamma_{\mathrm{min}}$ \\

$\gamma_{\mathrm{prod}}$
& $2\gamma_{\mathrm{prod}}^{1/2}$
& $K\gamma_{\mathrm{prod}}^{1/K}$ \quad or \quad 
$\dfrac{-K}{\ln \gamma_{\mathrm{prod}}}$ \\

$\gamma_{\mathrm{avg}}$
& $2\gamma_{\mathrm{avg}}$
& $2\gamma_{\mathrm{avg}}$ \\

$\gamma_{\mathrm{g.avg}}$
& $2\gamma_{\mathrm{g.avg}}$
& $e\gamma_{\mathrm{g.avg}}$ \\

\hline
\end{tabular}
\label{tab:RecDep}
\end{table}

\begin{figure}[t]
\begin{center}
\scalebox{0.5}{\includegraphics{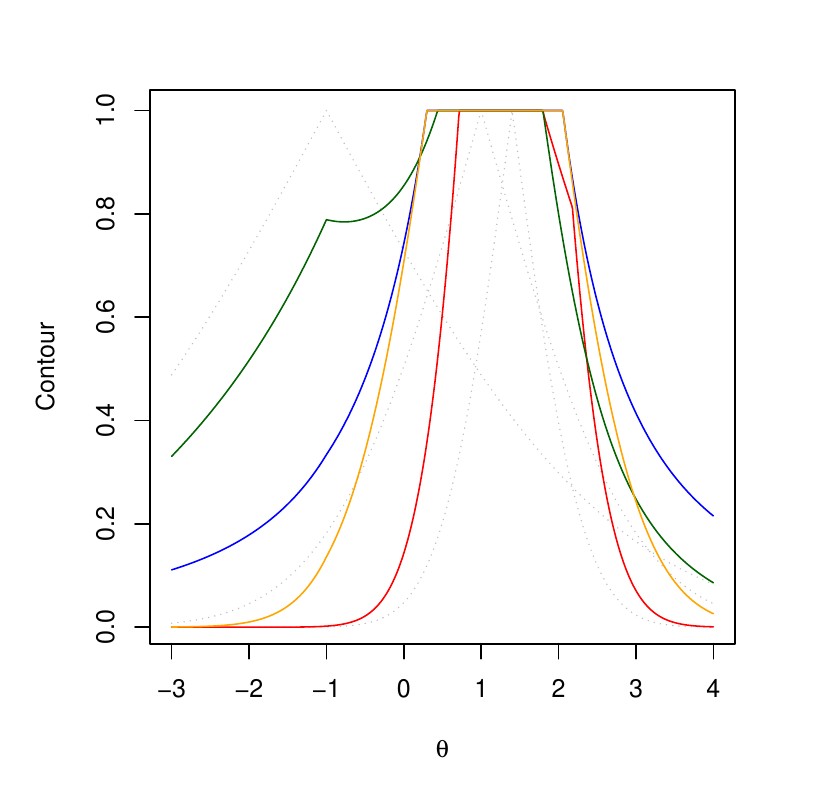}}
\end{center}
\caption{Fused IM contours obtained via the FVN strategy using the minimum (red), product (blue), average (green), and geometric average (orange) fusion operators, applied to the three dependent IM contours (grey).}
\label{fig:dep_fusion}
\end{figure}

\section{Fusing exchangeable IMs}
\label{s:exch}
We now consider an important intermediate case: the $K$ IMs are assumed to be exchangeable. Exchangeability occupies a spot somewhere in between independence and general dependence. While independence implies exchangeability, the converse does not hold. In fact, exchangeability can be viewed as a structured form of dependence, in which the joint distribution of the IM contours is invariant to permutations of their indices. This additional structure is still insufficient to determine the distribution of $\Gamma_f(\Theta)$, which prevents the ideal validification as done in Section~\ref{S:indep}. Fortunately, validification as in Section~\ref{s:dependent} remains feasible, and the following recent result introduced in \citet{ramdas2026randomized} suggests a path forward.

\begin{thm}[Exchangeable Markov Inequality]
    Let $X_1, \ldots, X_K$ be a set of exchangeable, non-negative and integrable random variables. Then, for any $a>0$,
    \[\mathbb{P}\left(\exists k\leq K : \frac{1}{k}\sum_{i=1}^kX_i\geq \frac{1}{a}\right) \leq a\mathbb{E}[X_1].\]
\end{thm}

This result suggests that the strategy in Section~\ref{s:dependent} of finding an equivalent event to $\Gamma_f(\Theta)\leq t$ that is suitable for the application of Markov’s inequality can be adapted here. However, solely considering $\Gamma_f(\Theta)\leq t$ with the knowledge that $\Pi_1(\Theta), \ldots, \Pi_K(\Theta)$ are exchangeable is not sufficient to derive events compatible with the exchangeable Markov inequality because it lacks the required ``$\exists k\le K$'' structure. This indicates that new ranking functions—different from simply applying fusion operators such as the minimum, product, linear-product, average, or geometric average—are required in the IM fusion step. 

As mentioned in Section~\ref{s:back}, the usual order of things in the IM framework is to first establish a principled ranking function and then rescale it through validification. However, situations in which the ease of validification guides the choice of ranking function are not new in the IM literature and can be motivated by efficiency and/or computational considerations; see, e.g., \citet{imdc}. In the present context, we need not look far. The ranking functions designed to allow validification using tools tailored to exchangeability can be obtained by suitably modifying the ones considered previously. 
Since the contours are exchangeable, it is natural to consider sequential prefixes $(\pi_1, \ldots, \pi_k), k=1, \ldots, K$,
 when constructing ranking functions.

More specifically, we propose the following combined ranking for $K$ exchangeable IM contours at a candidate value $\theta \in \TT$:
\[\gamma^{\mathrm{ex}}_f(\theta) = \bigwedge_{k=1}^K \gamma_f^{(1:k)}(\theta),\]
where $\gamma_f^{(1:k)}(\theta) = f(\pi_1(\theta), \ldots, \pi_k(\theta))$ and $\bigwedge_{k=1}^K y_k = \min\{y_1, \ldots, y_K\}$. In words, the combined ranking is obtained by applying the fusion operator $f$ sequentially to the first 
$k$ contours, $k=1, \ldots, K$, and then taking the minimum of the resulting values. The key point is that 
\[\Gamma^{\mathrm{ex}}_f(\theta) \leq t \quad \iff  \quad \exists k \leq K : \Gamma_f^{(1:k)}(\theta) \leq t,\]
so exchangeable Markov–compatible events can now be derived. Note, however, that for the minimum, product, and linear-product operators, the sequence of prefix values
$\gamma_f^{(1:1)}(\theta), \gamma_f^{(1:2)}(\theta), \dots, \gamma_f^{(1:K)}(\theta)$ is nonincreasing in $k$. Therefore,
$\bigwedge_{k=1}^K \gamma_f^{(1:k)}(\theta) = \gamma_f(\theta)$ for these operators, so that the resulting exchangeable Markov-compatible events reduce to the corresponding Markov-compatible ones obtained under arbitrary dependence, yielding the same upper bounds. Table~\ref{tab:ubex} reports the exchangeable Markov-compatible events and corresponding upper bounds for the average and geometric-average operators, for which the prefix construction yields genuinely new rankings. 

\begin{table}[tb]
\caption{Upper bounds for $\mathbb{P}_\Theta\{\Gamma^\mathrm{ex}_f(\Theta)\le t\}$, $t \in (0,1)$, based on the minimum bound and on exchangeable Markov's inequality.}
\renewcommand{\arraystretch}{1.3}
\setlength{\tabcolsep}{6pt}
\centering
\begin{tabular}{c c c c}
\hline
$\gamma_f^{\mathrm{ex}}$ & Min-based bound & $\Gamma_f^{\mathrm{ex}}(\Theta)\le t \iff \exists k \leq K :$ & Markov bound \\
\hline




$\gamma_{\text{avg}}^{\mathrm{ex}}$
& $K t$
& $\displaystyle \frac{1}{k}\sum_{j=1}^k 1-\Pi_j(\Theta)
   \ge 1-t$
& $\displaystyle \frac{1}{2(1-t)}$ \\[6pt]

$\gamma_{\text{g.avg}}^{\mathrm{ex}}$
& $K t$
& $\displaystyle \frac{1}{k}\sum_{j=1}^k -\ln \Pi_j(\Theta)
   \ge -\ln t$
& $\displaystyle \frac{-1}{\ln t}$ \\
\hline
\end{tabular}
\label{tab:ubex}
\end{table}


Analogous to the arbitrarily dependent case, the minimum bound can also be leveraged to obtain upper bounds for the average and geometric-average exchangeable rankings. These follow from
\[\Gamma_f^{\mathrm{ex}}(\Theta) \le t \implies \Gamma_\mathrm{min}^{\mathrm{ex}}(\Theta) \le t \iff \Gamma_{\mathrm{min}}(\Theta) \le t, \quad f\in\{\mathrm{avg},\mathrm{g.avg}\},\]
and are reported in Table~\ref{tab:ubex}. Comparing the exchangeable Markov-based bounds with the minimum-based ones leads to the same conclusions as in the arbitrarily dependent case: when $K=2$, the minimum-based bounds are uniformly tighter than the Markov bounds, whereas when $K>2$, neither bound uniformly dominates the other over $t\in(0,1)$. Importantly, the validified contours $\dot \gamma_f^{\mathrm{ex}}$ obtained from these bounds are at least as efficient as their $\dot \gamma_f$ counterparts under arbitrary dependence. 




P-value–based strategies were central in the arbitrarily dependent case, and
recent developments on combining exchangeable p-values motivate a similar
approach here. Recall from Section~\ref{s:dependent} that, under arbitrary dependence,
calibrators yield valid p-merging rules with rejection regions of the form
\eqref{eq:pvRR}. When $\boldsymbol{p}=(p_1,\ldots,p_K)$ consists of exchangeable p-values,
stronger constructions are available. In particular, \citet{RamdasExchangeable} show that, for any
calibrator $g$ and $\alpha \in (0,1)$, the rejection region
\begin{equation}\label{eq:exch_calibrator}
R_\alpha = \left\{
\boldsymbol{p}\in[0,1]^K :
\exists\, k \le K \text{ such that }
\frac{1}{k}\sum_{i=1}^k g\!\left(\frac{p_i}{\alpha}\right)\ge 1
\right\}
\end{equation}
defines an homogeneous p-merging rule under exchangeability. Note that this rejection region is larger than its arbitrarily dependent counterpart in \eqref{eq:pvRR},
and therefore has the potential to yield a more efficient merged p-value.

As in the arbitrarily dependent case, the key to leveraging these results in the IM context is the observation that
$\pi_1(\theta),\ldots,\pi_K(\theta)$ can be viewed as p-value analogues
for testing the assertion $\Theta=\theta$. This suggests seeking validified
exchangeable rankings whose rejection event $\dot\gamma^{\mathrm{ex}}_f(\theta)\le\alpha$, $\alpha \in (0,1)$,
can be expressed in the exchangeable calibrator form
\begin{equation}\label{eq:IM_exch_calibrator}
\exists\, k \le K \text{ such that }
\frac{1}{k}\sum_{i=1}^k g\!\left(\frac{\pi_i(\theta)}{\alpha}\right)\ge 1,
\end{equation}
for some calibrator $g$.

We again restrict attention to linear validifications of the form
$c\gamma_f^{\mathrm{ex}}$, with $c>0$, and seek values of $c$ for which the
rejection event
$c\gamma_f^{\mathrm{ex}}(\theta)\le\alpha$
admits a representation of the form \eqref{eq:IM_exch_calibrator} for some
calibrator $g$. This leads to the validified ranking
\[\dot\gamma_f^{\mathrm{ex}}(\theta)=\min\{1,c\gamma_f^{\mathrm{ex}}(\theta)\}.\]
The next theorem identifies such values of $c$ for $\gamma_{\mathrm{min}}^{\mathrm{ex}}$, $\gamma_{\mathrm{avg}}^{\mathrm{ex}}$, and
$\gamma_{\mathrm{g.avg}}^{\mathrm{ex}}$. Since $\gamma_{\mathrm{prod}}^{\mathrm{ex}}$ and
$\gamma_{\mathrm{l.prod}}^{\mathrm{ex}}$ are not homogeneous, they cannot be
validified within this p-value–based strategy. In fact, no linear validification is possible for these rankings, even under exchangeability. The proof is omitted,
as it follows from the same counterexample based on identical uniforms described in Theorem~\ref{thm:impossible-linear-validification}.

\begin{thm}
For the minimum, average, and geometric-average fusion operators, the choices
$c=K$, $c=2$, and $c=e$, respectively, allow the rejection event
$c\gamma_f^{\mathrm{ex}}(\theta)\le \alpha$, $\alpha \in (0,1)$, to be rewritten in the
exchangeable calibrator form \eqref{eq:IM_exch_calibrator}.
\end{thm}
\begin{proof}
Applying the same algebraic manipulations as in the arbitrarily dependent case
to each prefix $(\pi_1,\ldots,\pi_k)$ yields the following equivalences:
\begin{align*}
c \gamma_{\mathrm{min}}^{(1:k)}(\theta) \le \alpha
&\iff
\frac{1}{k}\sum_{i=1}^k K\mathbbm{1}\!\left\{
\frac{\pi_i(\theta)}{\alpha}\le \frac{1}{c}
\right\} \ge 1, \\
c \gamma_{\mathrm{avg}}^{(1:k)}(\theta) \le \alpha
&\iff
\frac{1}{k}\sum_{i=1}^k \left(c - c\frac{\pi_i(\theta)}{\alpha}\right)
\ge c-1, \\
c \gamma_{\mathrm{g.avg}}^{(1:k)}(\theta) \le \alpha
&\iff
\frac{1}{k}\sum_{i=1}^k -\ln\!\left(\frac{\pi_i(\theta)}{\alpha}\right)
\ge \ln c.
\end{align*}
Setting $c=K$, $c=2$, and $c=e$, respectively,
yield the desired representations in \eqref{eq:IM_exch_calibrator}, with the
same calibrators stated in the proof of Theorem~\ref{thm:CvalsDep}.
\end{proof}

Since $\gamma_{\mathrm{min}}^\mathrm{ex} = \gamma_{\mathrm{min}}$, the exchangeable p-value--based validification of the minimum operator coincides with its arbitrarily dependent counterpart. On the other hand, the resulting validified average- and geometric-average-based rankings are at least as efficient as those obtained from the p-value--based strategy under arbitrary dependence.

We end this section with a summary of the main conclusions and efficiency-based recommendations under exchangeability, which largely mirror those obtained under arbitrary dependence. In particular, the linear-product–based ranking remains unrecommended here, and the recommendations summarized in Table~\ref{tab:RecDep} continue to apply upon replacing $\gamma_f$ with $\gamma_f^{\mathrm{ex}}$. Potential efficiency gains relative to the arbitrarily dependent case arise only for $\dot\gamma_{\mathrm{avg}}^{\mathrm{ex}}$ and $\dot\gamma_{\mathrm{g.avg}}^{\mathrm{ex}}$, since $\dot\gamma_{\mathrm{min}}^{\mathrm{ex}}$, and $\dot\gamma_{\mathrm{prod}}^{\mathrm{ex}}$ coincide with their arbitrarily dependent counterparts. As an illustration, Figure~\ref{fig:compK3} displays the average- and geometric-average-based fused IM contours obtained under independence, arbitrary dependence, and exchangeability for the $K=3$ contours shown in Figure~\ref{fig:IM_contours}. The comparison highlights how additional structural assumptions can translate into gains in efficiency.


\begin{figure}[t]
\begin{center}
\subfigure[Average]{\scalebox{0.5}{\includegraphics{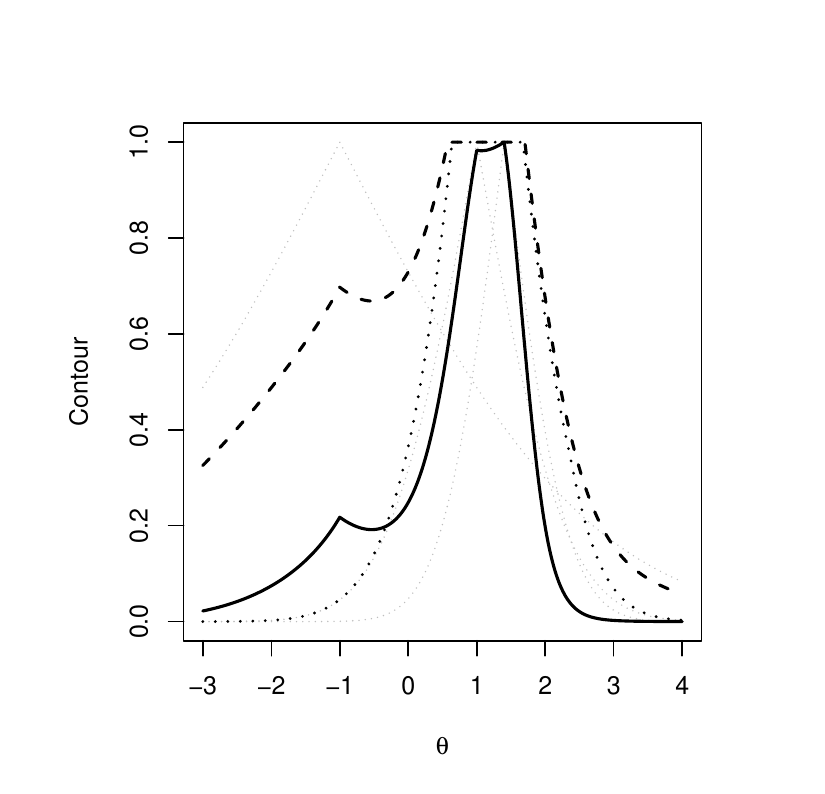}}}
\subfigure[Geometric Average]{\scalebox{0.5}{\includegraphics{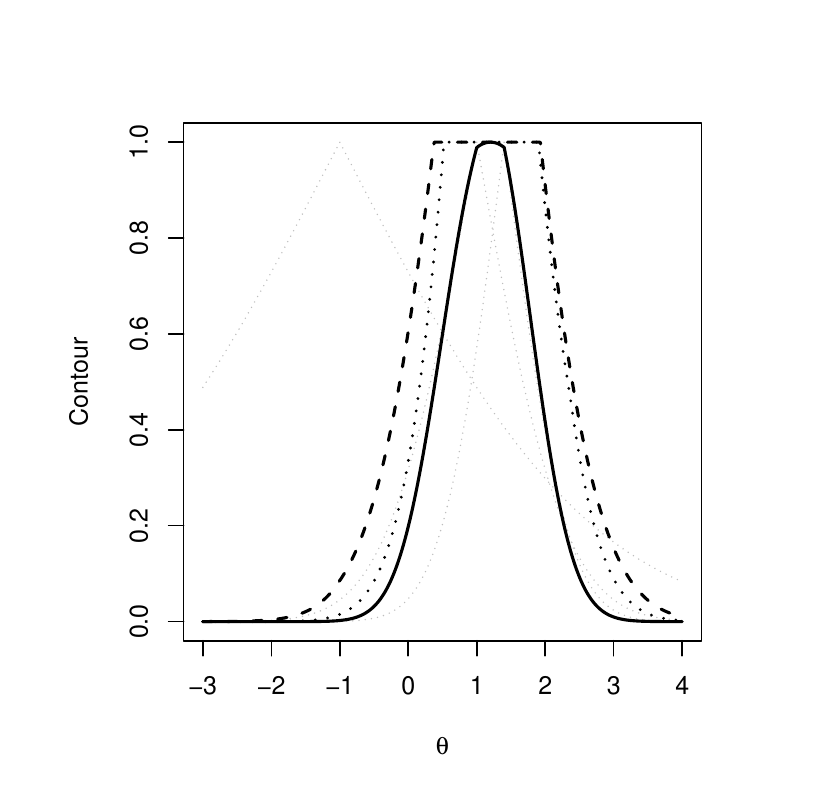}}}
\end{center}
\caption{Fused IM contours obtained via the FVN strategy applied to the three IM contours (grey) when assuming they are independent (solid), dependent (dashed), and exchangeable (dotted).}
\label{fig:compK3}
\end{figure}

\section{Examples}
\label{s:Examples}

\subsection{Meta-analysis}


We consider a simple hypothetical meta-analysis scenario in which $\Theta$ is the unknown rate of an exponential distribution and only the MLEs and sample sizes from $K$ independent studies are available. Owing to the regularity of the exponential model, the relative likelihood depends on the data only through this information, so the individual IM contours can be computed. Figure~\ref{fig:meta}(a) shows the case $K=3$, based on MLEs and sample sizes given by $\hat \theta_1=1.66$, $\hat \theta_2=0.91$, $\hat \theta_3=0.78$, and $n_1=3$, $n_2=6$, $n_3=9$.

Figure~\ref{fig:meta}(a) also shows the oracle fused IM that would be obtained if the full data from the three studies were available. In the exponential case, the full-data MLE and relative likelihood can be expressed as functions of the individual MLEs alone, so this oracle can be reconstructed from the available information. This, however, is rarely the case. More generally, the oracle IM based on the full data cannot be recovered from summary statistics alone. Fortunately, the individual contours are all that is needed to fuse independent IMs through the proposed FVN strategy. The fused IM based on the average fusion operator is shown in Figure~\ref{fig:meta}(b).

For comparison, Figure~\ref{fig:meta}(b) also shows a ``large-sample'' IM solution based on the contour
\[\pi(\theta) = 1 - F\bigl\{ J(\check\theta - \theta)^2 \bigr\}, \quad \theta \in \TT,\]
where $F$ is the $\chisq(1)$ distribution function,
$\check{\theta} = J^{-1}\sum_{k=1}^K n_{k} \hat\theta_k^{-1}$ ,
and $J = \sum_{k=1}^K n_{k} \hat\theta_k^{-2}$.
In words, this IM arises from treating each MLE as Gaussian with mean $\Theta$, which makes the weighted average $\check{\theta}$—with weights equal to the inverses of the observed variances, $n_k / \hat\theta_k^2$—the best linear unbiased estimator of $\Theta$ with variance $J^{-1}$. This yields an exact valid  IM when the MLEs are indeed Gaussian. When this is not the case, approximate validity holds when the sample size is sufficiently large; see \citet{imdc} for details. Note that this approach does not involve fusing individual IMs into a single IM, but it is nevertheless appealing in meta-analysis settings, since the IM can be constructed solely from the available summary statistics. From Figure~\ref{fig:meta}(b), the large-sample contour appears substantially more efficient than the average-based FVN one. However, this gain in efficiency should be interpreted with caution, since the former is not guaranteed to be valid, whereas the latter is. This is confirmed in a simulation study where the above scenario is repeated $10{,}000$ times, with datasets being generated from $\Theta = 0.8$. Figure~\ref{fig:meta}(c) shows the (empirical) distribution function of both contours at $\Theta=0.8$ in each Monte Carlo replicate. The large sample IM is stochastically smaller than \unif(0,1), indicating that it is not valid in the present setting. 

A legitimate criticism of the average-based FVN solution above is whether it achieves the best possible efficiency given all the available information. Because the FVN strategy relies solely on the individual contours, potentially informative quantities may be overlooked, depending on the chosen fusion operator. For example, when using the average fusion operator, each contour is assigned equal weight in the fusion process. This can be undesirable in settings such as meta-analysis, where additional information—such as sample sizes or observed variances—is available, suggesting that different contours should be weighted differently. This limitation is not unique to the average; the minimum, product, and geometric average operators share the same feature. These considerations highlight the need to look beyond the standard fusion operators commonly used in the possibility literature when fusing IM contours. 
Here, a weighted average fusion operator with weights proportional to the inverse variances is a natural choice. Its distribution for validification has no closed form because of the data-dependent weights, 
so we approximate it through Monte Carlo.
Figure~\ref{fig:meta}(b) shows the resulting contour, which is strictly more efficient than the average-based one. Its validity is confirmed in Figure~\ref{fig:meta}(c) using the same simulation study described above. To verify that the observed efficiency gain is not specific to a single dataset, Table~\ref{tab:dp} reports the average widths of the confidence intervals obtained from the fused IMs for various $\alpha$ levels across the Monte Carlo replicates. The efficiency gain from the weighted-average fusion is clear. All associated standard errors were below $0.01$.


\begin{figure}[t]
\begin{center}
\subfigure[]{\scalebox{0.37}{\includegraphics{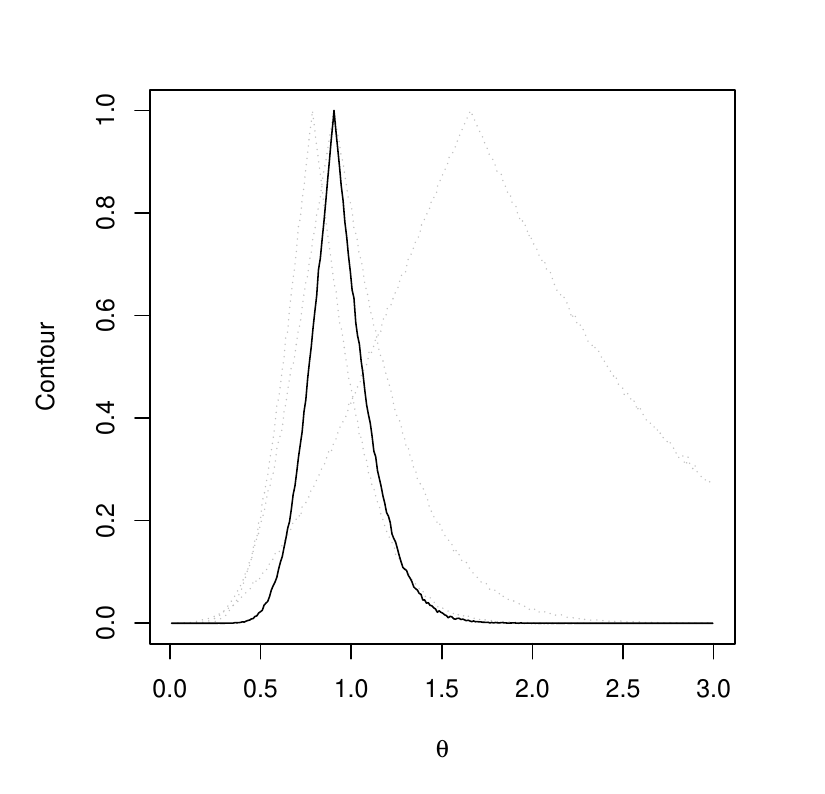}}}
\subfigure[]{\scalebox{0.37}{\includegraphics{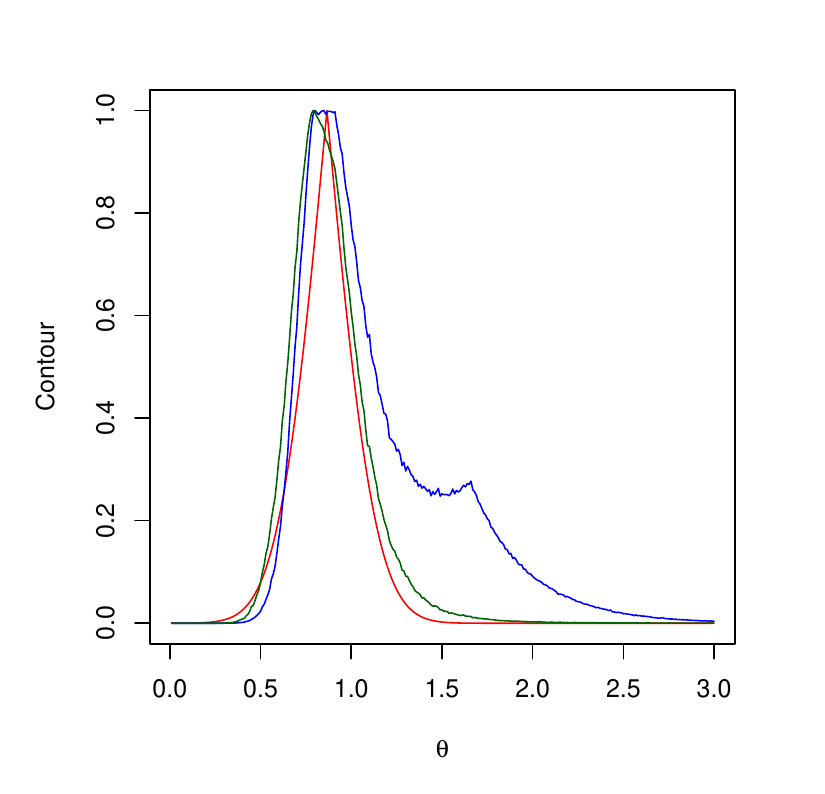}}}
\subfigure[]{\scalebox{0.37}{\includegraphics{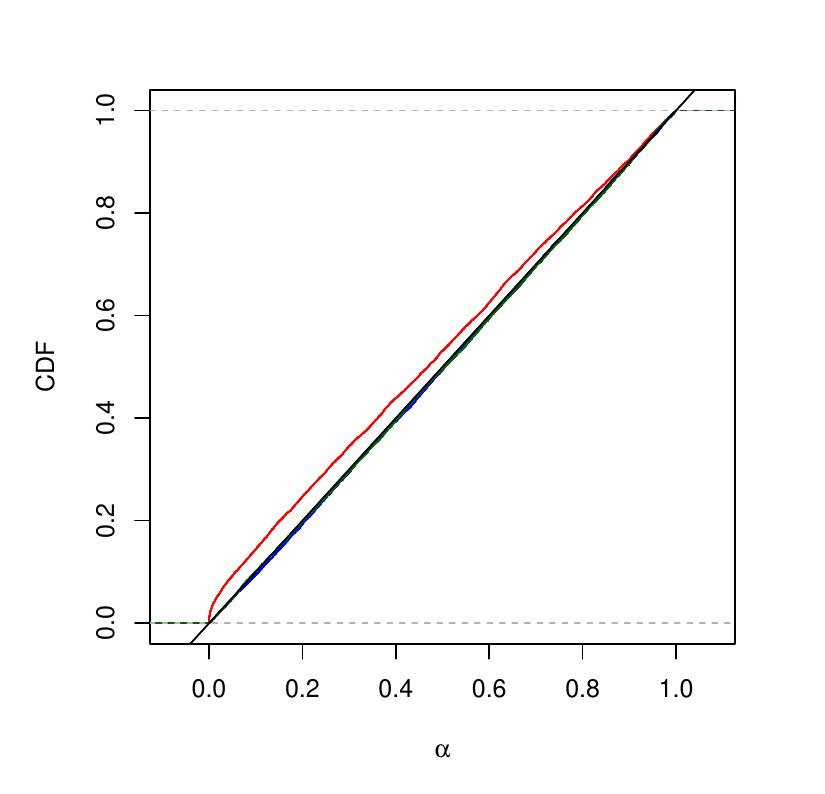}}}
\end{center}
\caption{
Panel~(a) shows three IM contours (grey) and the oracle IM contour (black). Contours and empirical distribution functions related to the large sample IM (red), average-based fused IM (blue) and weighted average fused IM (green) are shown in Panels (b) and (c), respectively.
}
\label{fig:meta}
\end{figure}

\begin{table}[tb]
\caption{Average interval widths.} \label{tab:dp}
\centering
\begin{tabular}{lrrrrr}
\hline
\multirow{2}{*}{Fusion operator} & \multicolumn{5}{c}{$100(1-\alpha)\%$} \\
 & $10$ & $30$ & $50$ & $70$ & $90$ \\
\hline
average & 0.17 & 0.38 & 0.59 & 0.86  & 1.26 \\
weighted-average  & 0.11 & 0.25  & 0.39 & 0.55 & 0.83 \\
\hline
\end{tabular}
\end{table}

\subsection{Distribution-free inference on the conditional median}

Consider the problem of quantifying uncertainty about the median of a response variable at a new covariate value, after observing $n$ iid covariate–response pairs, without imposing any parametric assumptions on the data-generating mechanism. Inspired by the developments in \citet{CandesConditionalMedian}, \citet{cella2025} propose an IM construction that relies on a sample-splitting strategy. One portion of the data is used to build a preliminary estimate of the conditional median, which is then used to carry out the IM construction on the remaining data. This leads to a distribution-free IM for the conditional median that enjoys marginal calibration guarantees.

Consider the dataset with $n = 20$ shown in Figure~\ref{fig:CondMed}(a), where $Y$ is the response variable and $X$ is the covariate. Denote by $m(x)$ the median of the conditional distribution of $Y$ given $X = x$. To apply the distribution-free IM construction for $\Theta = m(x_{n+1})$ described above, we first randomly split the data into two halves, with the first half represented by the black data points in Figure~\ref{fig:CondMed}(a). The curve in the same panel corresponds to the second-degree median regression model used to estimate $m(x)$ based on these data. Suppose $X_{n+1} = 3$. Figure~\ref{fig:CondMed}(b) displays the resulting IM contour constructed from the remaining half of the data, based on this estimate.

The contour shown in Figure~\ref{fig:CondMed}(b) corresponds to a single realization of the random sample-splitting procedure. Repeating this procedure produces multiple IM contours for $\Theta$. Figure~\ref{fig:CondMed}(c) shows nine additional contours, revealing substantial variability across splits and highlighting the potential benefit of aggregation. Since all contours arise from the same randomized mechanism, they are exchangeable, which justifies the use of the exchangeability-based fusion strategies developed in Section~\ref{s:exch}. The average-based fused IM is shown in Figure~\ref{fig:CondMed}(c). For comparison, the average-based fused IM under arbitrary dependence is also displayed, illustrating the efficiency gains achieved by exploiting exchangeability. 

To illustrate a scenario in which relying on standard fusion operators from the possibility literature may be detrimental from an efficiency standpoint, we repeat the above analysis with a dataset of size $n = 500$. As shown in Figure~\ref{fig:CondMed}(d), the contours resulting from repeated splitting are now much more homogeneous, and fusion based on the average operator leads to a noticeable loss in efficiency. This behavior is not specific to the average; the same phenomenon occurs for the minimum, product, and geometric average operators, although we omit the corresponding contours from Figure~\ref{fig:CondMed}(d). Interestingly, in this setting, the maximum fusion operator,
$\gamma_{\mathrm{max}}(\theta) = \max\{\pi_1(\theta), \ldots, \pi_K(\theta)\}$,
provides a more efficient alternative. Figure~\ref{fig:CondMed}(d) shows that the fused contour under arbitrary dependence—coinciding with $\gamma_{\mathrm{max}}(\theta)$, since it is already valid and normalized—achieves the desired stabilization while achieving higher efficiency. Although the maximum operator is classified as disjunctive in the possibility literature and is typically recommended in situations where at least one contour may be unreliable (which is not the case here), its effectiveness for IM fusion in this setting highlights that intuition based on the possibility-fusion literature does not always carry over when validity and/or efficiency are the primary considerations. Figure~\ref{fig:CondMed}(d) also displays the fused IM under the exchangeability assumption, based on
\[
\gamma^{\mathrm{ex}}_{\mathrm{max}}(\theta) = \bigwedge_{k=1}^K \gamma_{\mathrm{max}}^{(1:k)}(\theta).
\]
Since these prefix maxima are nondecreasing in $k$, it follows that $\gamma^{\mathrm{ex}}_{\mathrm{max}}(\theta) = \pi_1(\theta)$, which is again already valid and normalized, and evidently more efficient than $\gamma_{\mathrm{max}}$. This supports the intuition that, when repeated splitting and aggregation are detrimental from an efficiency perspective, a single split may be preferable.

\begin{figure}[tbh]
\begin{center}
\subfigure[]{\scalebox{0.43}{\includegraphics{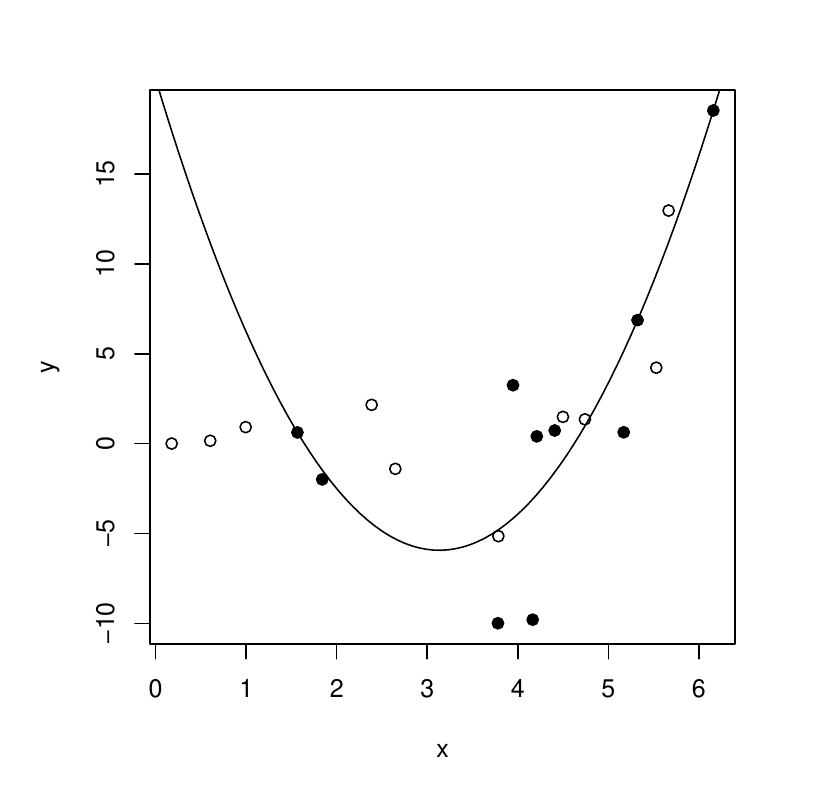}}}
\subfigure[]{\scalebox{0.43}{\includegraphics{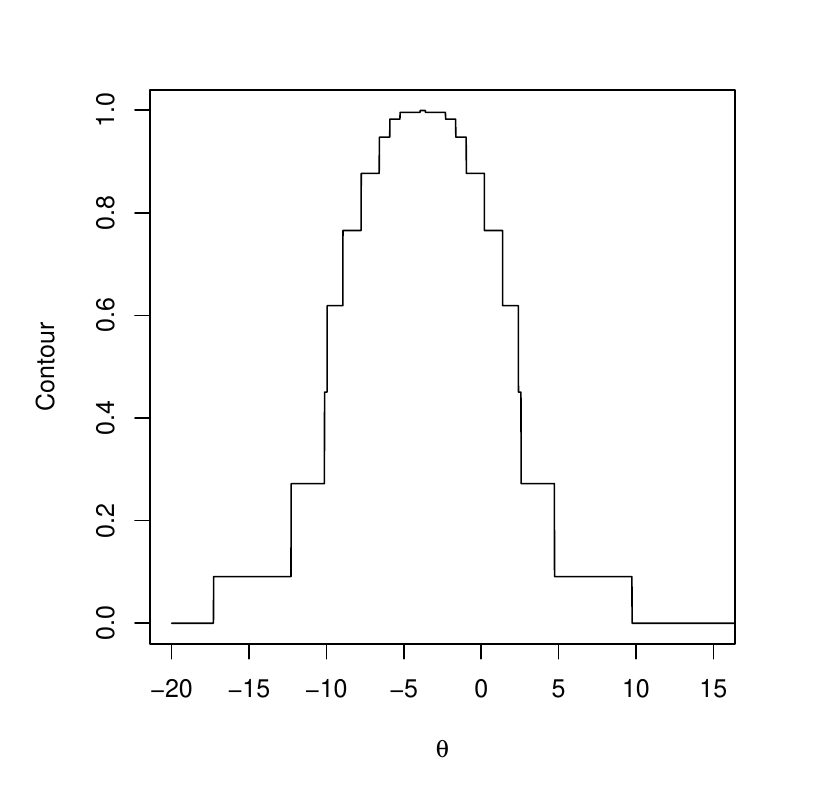}}}
\subfigure[]{\scalebox{0.43}{\includegraphics{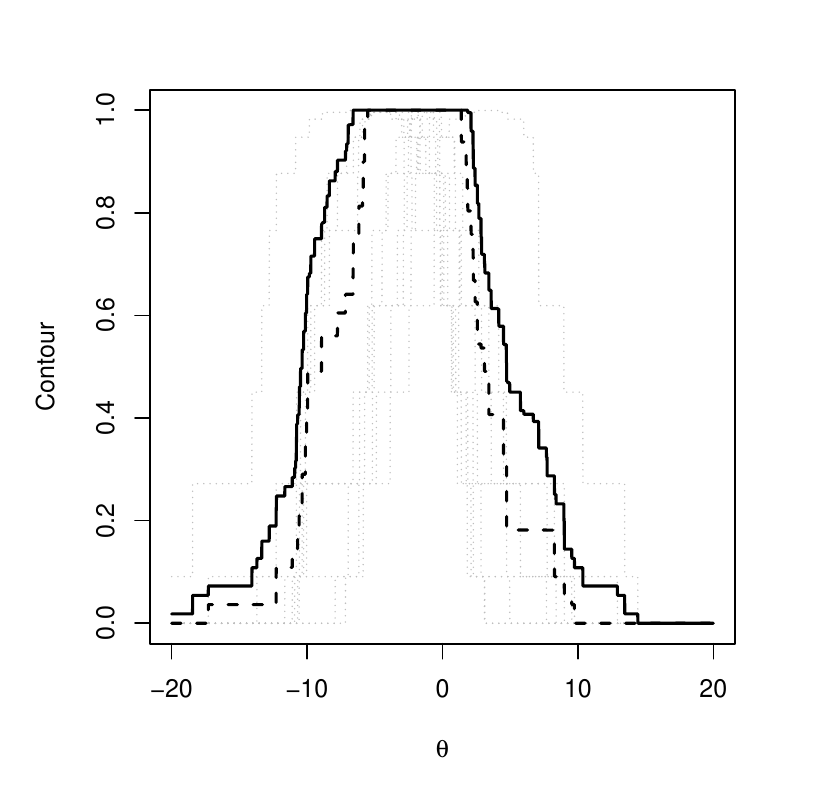}}}
\subfigure[]{\scalebox{0.43}{\includegraphics{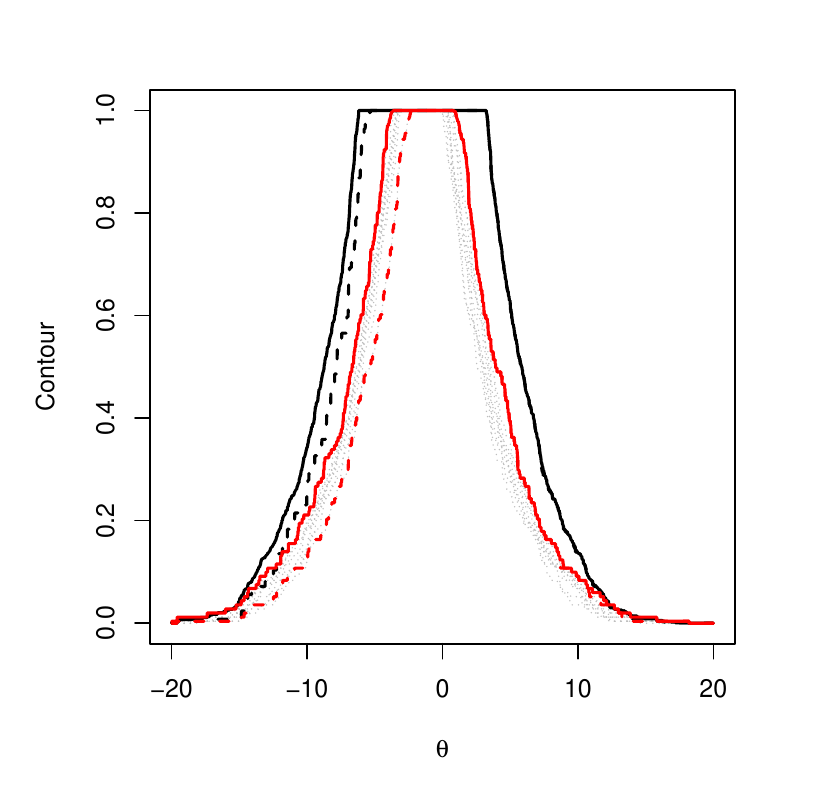}}}
\end{center}
\caption{Panel~(a) shows an observed dataset, a random split, and the corresponding regression fit based on the black points. Panel~(b) displays the resulting IM contour. Panel~(c) presents 10 contours (grey) obtained from different random splits of the dataset in Panel~(a), together with average-based fused IMs under arbitrary dependence (solid black) and exchangeability (dashed black). Panel~(d) presents the analogous experiment for a dataset of size $n=500$, additionally including maximum-based fused IMs (red).
}
\label{fig:CondMed}
\end{figure}

\section{Conclusion}
\label{s:Conclusion}

This paper introduces a general framework for validity-preserving fusion of possibility contours, a fundamental development for the continued growth of the IMs approach to statistical inference. The underlying logic of the proposed framework mirrors the core rationale of the IM framework itself, making it akin to an application of IM methodology to settings where the ``data'' consist of IM possibility contours. The framework is intentionally general and was studied under independence, arbitrary dependence, and exchangeability of the available IMs, thereby equipping IM users with principled fusion strategies across a broad range of both classical and modern inferential settings in which combining IMs may be required. Beyond the theoretical developments, the paper also provides concrete efficiency-based recommendations for the validification of commonly used fusion operators under the different dependence structures considered.

Once validity and efficiency become part of the fusion problem, selecting a fusion operator solely based on interpretation---as is typically done in the possibility literature---is no longer the end of the story, and several forms of unexpected behavior can arise. Interpretation-wise different fusion operators can lead to the same fused outcome, which was the case for the product and geometric-average fusion operators under independence. Popular fusion operators from the possibility literature may become entirely unrecommended, as happened with the linear-product rule. It can also be advantageous to move beyond the traditional fusion operators from the possibility literature in pursuit of greater efficiency, especially when additional structural information is available, as illustrated by the weighted-average in a meta-analysis setting. More surprisingly, efficiency considerations may i) favor fusion strategies that run counter to the usual intuition from the possibility literature, such as the use of the maximum, a disjunctive-type fusion operator, in settings where conjunctive or statistical fusion would traditionally appear more natural; and ii) require more elaborate fusion mechanisms that go beyond the mere application of a fusion operator, as illustrated by the exchangeable setting.

We conclude with two directions for future research. First, the normalization in \eqref{eq:normGeneral}, while intuitive and simple in principle, can become computationally challenging in high-dimensional settings. Developing efficient strategies for computing this normalization therefore remains an important open problem. Second, ideas from the p-value merging literature proved highly useful for the problem of fusing IMs. Recent e-value--based developments on anytime-valid and post-hoc p-values (see \citet{RamdasEVbook} and the references therein) suggest a clear direction for future work on IM fusion with additional guarantees. In this regard, the recent ``e-possibilistic IM'' framework of \citet{martin2025regularized} provides a promising starting point.


\section*{Acknowledgments}
The author thanks Professors Ryan Martin and Aaditya Ramdas for helpful discussions and valuable suggestions.

\bibliographystyle{apalike}
\bibliography{Bibli.bib}

\end{document}